\documentclass[letterpaper]{article} 
\usepackage[preprint]{aaai2027}  
\usepackage[hyphens]{url}  
\usepackage{graphicx} 
\usepackage{natbib}  
\usepackage{caption} 
\usepackage{amsmath}
\usepackage{amssymb}
\usepackage{hyperref}

\usepackage{amsthm}

\newtheorem{theorem}{Theorem} 
\newtheorem{lemma}[theorem]{Lemma}

\newtheorem{definition}{Definition}
\newtheorem{assumption}{Assumption}

\newcommand{\norm}[1]{\lVert #1 \rVert}

\usepackage{algpseudocode}
\usepackage{algorithm}

\usepackage{newfloat}
\usepackage{listings}
\DeclareCaptionStyle{ruled}{labelfont=normalfont,labelsep=colon,strut=off} 
\floatstyle{ruled}
\newfloat{listing}{tb}{lst}{}
\floatname{listing}{Listing}

\usepackage{booktabs}

\title{Learning to Persuade Privately Informed Receivers}
\author{
    I. Arda Vurankaya,
    Ufuk Topcu
}
\affiliations{
    The University of Texas at Austin
}

\hypersetup{
    pdfauthor={I. Arda Vurankaya, Ufuk Topcu},
    pdftitle={Learning to Persuade Privately Informed Receivers},
    pdfsubject={},
    pdfkeywords={},
    hidelinks
}

\begin{document}

\maketitle

\begin{abstract}
Bayesian persuasion studies how an informed sender can influence the behavior of a receiver through strategic information disclosure. Standard models assume the sender is the receiver's only source of information, yet in many applications receivers also consult external sources the sender can neither observe nor control. We study an online Bayesian persuasion problem in which a binary-action receiver has access to a \emph{fixed} signaling scheme that is \emph{unknown} to the sender. Over $T$ rounds, the sender commits to a signaling scheme and sends a signal; the receiver combines it with its private signal and acts, while the sender observes only the action. We design a learning algorithm that achieves regret $\widetilde{O}(T^{3/4})$ relative to the optimal scheme of a sender who knows the private signaling scheme of the receiver, with polynomial dependence on the sizes of the state space and the receiver's signal alphabet. Our key insight is reducing the problem of learning the exponentially large belief-space partitioning induced by the private scheme to a one-dimensional change-point detection problem.
\end{abstract}


\section{Introduction}

Bayesian persuasion, introduced by \citet{kamenica2011bayesian}, is a model of how an informed sender can strategically disclose information to influence the behavior of decision makers. Since its introduction, it has found applications in many fields, such as online advertising \citep{emek2014signaling}, security \citep{rabinovich2015information}, traffic routing \citep{das2017reducing}, and incentivizing exploration \citep{mansour2016bayesian}. 

In its simplest form, Bayesian persuasion involves two agents: a sender who can observe the state of the world, and a receiver who has to take an action. Both of the agents' utilities depend on the state of the world and the action taken by the receiver. However, the agents' utilities are not perfectly aligned. The sender's goal is to selectively disclose information by publicly committing to a \emph{signaling scheme}, so as to influence the receiver to take a favorable action. 

In many applications of Bayesian persuasion, the receiver does not rely on the sender alone: before acting, it also observes signals about the state of the world from external sources that the sender can neither see nor control. The sender does not know what these sources reveal, and can only detect their effect indirectly, through the actions the receiver takes. This stands in contrast to the most commonly studied model, in which the sender is the receiver's sole information provider. As an example, consider an online marketing platform interacting with a stream of customers. In each interaction the platform observes the product's quality and chooses what to reveal, but customers also consult a common third-party review site before purchasing. The platform never sees what the customer read; it observes only the final decision.

\paragraph{Contributions.} We study an online variant of Bayesian persuasion with a binary-action receiver, where the receiver has access to an \emph{unknown} but \emph{fixed} signaling scheme. To the best of our knowledge, learning to persuade in the presence of an external, unknown signaling scheme has not been studied before.  The interaction between the sender and the receiver proceeds for multiple rounds. At each round, the sender commits to a signaling scheme, observes the state of the world and sends a signal to the receiver accordingly. The receiver observes an additional, conditionally independent signal from a signaling scheme unknown to the sender,  and forms its final belief based on the two signals. After each round, the sender only observes the action taken by the receiver. The sender's goal is to minimize its \emph{regret}, which is defined as the difference between the cumulative expected payoff the sender actually achieves and the payoff it could have achieved if it had full knowledge of the unknown signaling scheme.

We provide a learning algorithm that achieves, by using an explore-then-commit strategy, an expected
regret $\widetilde O(T^{3/4})$ with polynomial dependence on the sizes of the state space and the receiver's signaling alphabet, 
under regularity assumptions on the unknown signaling scheme and the receiver's utility. Our regret bound is information-theoretic: the number of exploration rounds is polynomial in problem parameters, but computing the final commitment takes exponential time.  The main idea behind the algorithm is observing how the \emph{state-conditional}  action distribution of the receiver changes when the sender's signal induces different posteriors. Intuitively, this conditional action distribution captures the additional information available to the receiver beyond what the sender reveals. We show that the problem reduces to learning a partitioning of the simplex by hyperplanes, and this learning problem can be reduced to identifying the jump points of a monotone, piecewise constant function, which allows us to use a recursive binary search procedure. Once the unknown signal likelihoods and hyperplanes are learned, the sender commits to a robust signaling scheme. Our analysis carefully controls the suboptimality in the final commitment in terms of the cost of robustness and other estimation errors.   

\paragraph{Related Work.}
Existing work on persuasion with privately informed receivers largely assumes that the receiver's information structure is known. When private information corresponds to the receiver's utility, \citet{kolotilin2017persuasion,kolotilin2018optimal} characterize optimal signaling schemes for a binary-action receiver. Closer to our model, \citet{hossain2024multi} study receivers who observe signals correlated with the state; a sender's best-response problem in their multi-sender model coincides with the offline, full-information version of ours. In contrast, we study how the sender can learn to persuade when the receiver's signaling scheme is initially unknown.

A separate literature studies online persuasion under uncertainty about the receiver's type. \citet{castiglioni2020online} assume that the finite set of possible types and their associated utility functions are known, but the sender does not know which adversarially selected type it faces in each round; \citet{bernasconi2023optimal} establish optimal regret rates for this setting, and \citet{castiglioni2021multi} extend it to multiple receivers. Other works instead consider unknown fixed primitives: \citet{bacchiocchi2024online} study unknown priors and receiver utilities, while \citet{li2025information} focus on an unknown prior. In all these settings, however, the sender remains the receiver's only source of information about the state. In our setting, an unobserved exogenous signal selects a latent best-response rule in each round, while the unknown private signaling scheme determines both the collection of such rules and their state-dependent occurrence probabilities. Hence, unlike the type-uncertainty setting, neither the active rule nor the underlying information structure is known to the sender.

In terms of what the learner initially knows and the feedback structure, our model is also related to Bayesian Stackelberg games with unknown follower types. \citet{bollini2026learning} show that no-regret learning is generally impossible under action feedback. Their impossibility result does not directly cover our information-design setting: although the sender does not observe the receiver's private signal, it observes the state, which is correlated with that signal. We exploit this additional structure, together with the geometry of binary-action persuasion, to obtain no-regret learning.

\section{Preliminaries}

\paragraph{Notation.} For a set $\mathcal{S} $, $\Delta(\mathcal{S}) $ denotes the set of all probability distributions defined on $\mathcal{S} $. When $\mathcal{S} $ is finite, we define $\Delta_{\geq x}(\mathcal{S}) := \{\mathbb P \in \Delta(\mathcal{S}): \mathbb P(s) \geq x \text{ } \forall s \in \mathcal{S}\}$. We denote by $\text{relint}(\mathcal{S}) $ the relative interior of $\mathcal{S} $. The cardinality of $\mathcal{S} $ is  $|\mathcal{S} | $. We denote by $[k] $ the set of integers $\{1, \ldots, k \} $.

\subsection{Standard Bayesian Persuasion}

A Bayesian persuasion instance with a binary action receiver consists of a finite state space \(\Theta \), an action set \(A = \{a_0, a_1\} \), a prior \(\mu_0 \in \text{relint}( \Delta(\Theta))\), and utility functions $u,u^R : A \times \Theta \to [0,1]$, where $u $ and $u^R $ are the sender's and the receiver's utilities, respectively. Before the state is realized, the sender commits to a signaling scheme $\phi:\Theta \to \Delta(S)$, where \(S\) is a finite signal alphabet.

The interaction proceeds as follows. First, the sender commits to a signaling scheme \(\phi\). Then a state \(\theta \sim \mu_0(\cdot)\) is drawn and observed by the sender. The sender draws a signal \(s \sim \phi(\cdot \mid \theta)\) and sends it to the receiver. Upon observing \(s\), the receiver forms the posterior $\mu_s $ by performing a Bayesian update and chooses an action
\[
a \in  \arg\max_{a' \in A}
\sum_{\theta' \in \Theta} \mu_s(\theta') u^R(a',\theta').
\]
The sender and the receiver obtain payoffs $u(a,\theta) $ and $u^R(a,\theta) $, respectively. We assume that the receiver breaks ties in favor of the sender, which is a standard assumption in the literature \citep{dughmi2016algorithmic}, and we denote the action chosen by the receiver at a posterior $\mu $ as $a^*(\mu) $. We denote by $\Theta_i $ the set of states at which the receiver weakly prefers $a_i $ for $i \in \{0,1\} $. 

The utility of a signaling scheme $\phi $ to the sender is  
\begin{equation}
    V(\phi) = \sum_{\theta \in \Theta} \sum_{s \in S}\mu_0(\theta) \phi(s|\theta)u(a^*(\mu_s), \theta),
\end{equation}
and the sender's problem is computing a $\phi^* \in \arg \max_{\phi} V(\phi) $. We finally remark that a signaling scheme $\phi $ can be equivalently expressed in terms of the set of posteriors induced by it as  $ \{\pi_s, \mu_s\}_{s \in S} $, where $\pi_s := \sum_{\theta}\mu_0(\theta)\phi(s|\theta) $ is the probability that $s $ is sent, and we have $\mu_0 = \sum_s \pi_s \mu_s $, which is the so-called Bayes-plausibility condition. 

\paragraph{Note.} We use $u_{\theta} (a) $ and $\phi_{\theta} (s) $ interchangeably with $u(a, \theta) $ and $\phi(s \mid \theta) $, respectively, in the rest of the paper. 

\subsection{Persuasion with a Privately Informed Receiver}

In contrast to the standard Bayesian persuasion in which the sender has full control over the receiver's information structure, we consider a receiver who has access to another signaling scheme, denoted by $\phi^R: \Theta \rightarrow \Delta(S^R)  $. After receiving signals $s \sim \phi(\cdot|\theta) $ and $r \sim \phi^R(\cdot |\theta) $, the receiver performs a Bayesian update as follows:
\begin{equation} \label{eq: Bayes}
    \mu_{s,r}(\theta) = \frac{\phi(s|\theta)\phi^R(r|\theta)\mu_0(\theta)}{\sum_{\theta'} \phi(s|\theta')\phi^R(r|\theta')\mu_0(\theta')},
\end{equation}
and takes the action $ a^*(\mu_{s,r}) $. For ease of notation, we denote $a^*(\mu_{s,r}) $ as $a^*(s,r) $ in the rest of the paper. Note that we assume $s $ and $r $ are conditionally independent given $\theta $. 

When the sender commits to $\phi $ and the receiver has access to $\phi^R $, the receiver can be thought to have access to a combined signaling scheme, denoted $(\phi, \phi^R) $, with signal alphabet $S \times S^R $ and likelihood values $\phi(s|\theta)\phi^R(r|\theta) $. In this case, the sender's problem is to compute a $\phi^* \in \arg \max_{\phi} V(\left (\phi, \phi^R)\right) $. An optimal $\phi $ can be computed as a solution to the following mixed-integer program \citep{hossain2024multi}:
\begin{subequations}
\label{eq:persuasion}
\begin{align}
\max_{\phi,y}\quad
& \sum_{\theta,s,r}
  \mu_0(\theta)\phi_\theta(s)
  \phi^R_\theta(r)
  u_\theta(a_{y_{s,r}} )
\label{eq:compute-obj}
\\
\text{s.t.}\quad
& \phi_\theta \in \Delta(S),
\quad  \theta\in\Theta,
\label{eq:compute-simplex}
\\
& y_{s,r} \in \{0,1\} ,
\quad (s,r)\in S\times S^R,
\label{eq:compute-one-action}
\\
& \sum_{\theta \in \Theta}
   \mu_0(\theta)\,
   \phi_{\theta} (s )\, 
   n_r(\theta)(-1)^{y_{s,r}} \geq 0, \,
   \label{eq:incentive}
\end{align}
\end{subequations}
where $n_r \in \mathbb{R}^{|\Theta|} $ is the unit vector in the same direction as the vector $\{ \phi^R_{\theta}(r)\left (u^R(a_0, \theta) - u^R(a_1, \theta) \right ) \}_{\theta} $. In the above program, $y_{s,r} $ specifies the action the receiver takes when it observes signals $s $ and $r $ from the sender and the private source, respectively. The constraint \eqref{eq:incentive} is an incentive-compatibility condition, making sure the receiver indeed takes the action specified by $y_{s,r} $. We finally remark that, by the standard concavification argument \citep{kamenica2011bayesian}, there exists an optimal scheme using at most $|\Theta|$ signals.

\subsection{Learning Problem}

In this paper, we consider how a sender can learn to optimally persuade a privately informed binary-action receiver through repeated interactions, when the sender does not have knowledge of $\phi^R $ initially. Hence there is uncertainty about the objective function \eqref{eq:compute-obj} and the feasible set \eqref{eq:incentive}. The sender, however, knows the prior $\mu_0 $, utility functions $u$ and $u^R$, and the size of $S^R $.  

The sender and the receiver interact for $T $ rounds. At each round $t \leq T $, the interaction protocol is as follows:
\begin{enumerate}
    \item The sender publicly commits to a signaling scheme $\phi_t $. 
    \item The state of the world $\theta_t $ is drawn from $\mu_0(\cdot) $.

    \item The sender observes $\theta_t $ and sends $s_t \sim \phi_t(\cdot |\theta_t). $

    \item The receiver observes $s_t  $ and $r_t \sim\phi^R(\cdot|\theta_t) $.

    \item The receiver performs a Bayesian update as in \eqref{eq: Bayes}, and takes the action $a_t = a^*(s_t, r_t) $.   
\end{enumerate}
Let $V^* = \max_{\phi}V((\phi, \phi^R)) $. We define the regret as
\begin{equation}
    R_T = V^*T- \mathbb{E} \left [ \sum_{t =1}^T V(\phi_t, \phi^R)  \right ], 
\end{equation}
where the expectation is with respect to the randomness of the algorithm. Our goal is to design a learning algorithm for the sender such that the regret is sublinear, i.e., $R_T = o(T) $.

\section{The Geometry of Persuading a Privately Informed Receiver}

In this section, we discuss the geometric properties of the problem of persuading a privately informed receiver, which allows us to design a sample-efficient no-regret learning algorithm. 

When the sender induces a posterior $\mu_s \in \Delta(\Theta) $ by sending a signal $s $, the action taken by the receiver depends on which signal $r \in S^R $ the receiver observes. Since $r $ and $s $ are independent given the state $\theta ,$ the receiver takes the action 
\begin{equation}
    a^*(s,r) \in \arg \max_a \sum_{\theta} \mu_s(\theta)\phi^R(r|\theta)u^R(a, \theta),
\end{equation}
This corresponds to a partitioning of $\Delta(\Theta) $ into two regions by a hyperplane $H_r $, with normal vector given by $n_r $.

For all $\mu_s \in \Delta(\Theta) $
 we define the corresponding \emph{state-conditional action distribution}
\begin{equation}
    p_{\mu}(a_0|\theta) = \sum_{r: a^*(s,r)=a_0}\phi^R(r|\theta).
\end{equation}
We define the vector $p_{\mu} \in [0,1]^{|\Theta|} $ such that $p_{\mu}(\theta) = p_{\mu}(a_0|\theta) $. It follows that the set of hyperplanes $\{H_r\}_{r \in S^R} $ forms a partitioning of $\Delta(\Theta) $ into polyhedral cells such that, within each region, $p_{\mu} $ stays constant. Equivalently, each cell $C $ can be associated with a mapping $w_C:S^R \rightarrow A $, where $w_C(r) $ is the action taken by the receiver when the sender's signal induces a posterior $\mu \in C $ and the receiver observes $r \in S^R $. We show in Figure \ref{fig:geometry} a comparison between standard persuasion and our setting in an example with $|\Theta |=3 $ and $|S^R|=3 $.

\begin{figure}[t]
    \centering
    \includegraphics[width=\linewidth]{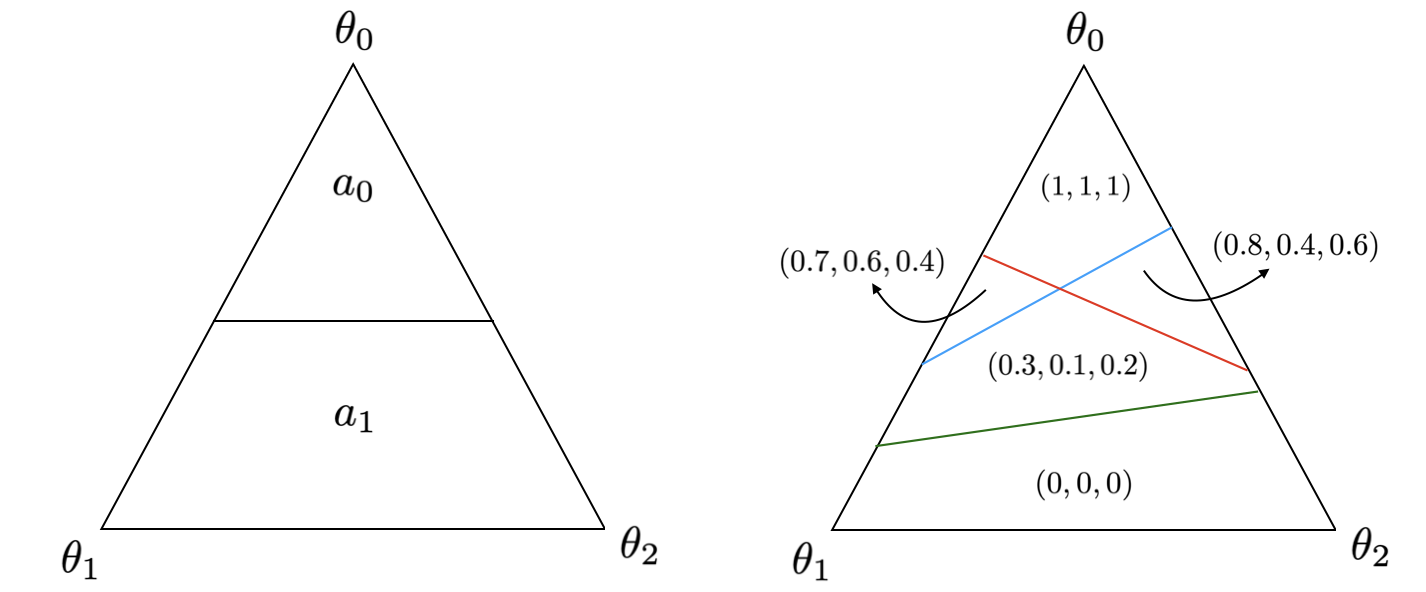}
    \caption{Comparison between the geometry of standard binary-action Bayesian persuasion and Bayesian persuasion with a privately informed receiver, where $|S^R|=3 $. The numerical values shown correspond to $p_{\mu}(a_0|\theta) $. Each hyperplane $H_r $ corresponding to an $r \in S^R $ results in a jump in $p_{\mu} $, where the jump magnitude is given by $\phi^R(r|\theta) $. }
    \label{fig:geometry}
\end{figure}

\section{Learning to Persuade a Privately Informed Receiver}

In this section, we design our learning algorithm. As mentioned in the previous section, the receiver's private signaling scheme $\phi^R $ partitions $\Delta(\Theta) $ based on $p_{\mu} $. Our algorithm learns this partitioning by identifying how $p_{\mu} $ changes across $\Delta(\Theta) $. Efficiently implementing this approach presents several challenges.

The \textbf{first} challenge the algorithm needs to overcome is the fact that the number of polyhedral regions in the partitioning can be as large as $O(|S^R|^{|\Theta|}) $ \citep{xu2020tractability}. A naive approach to the learning problem we study is to learn all such polyhedral regions in isolation, by adopting previous approaches to learning in standard Bayesian persuasion \citep{bacchiocchi2024online}. We instead show that, using the geometry of Bayesian updating and binary action structure, it is enough to learn at most $|S^R| $ hyperplanes, and that this can be done in a number of exploration rounds that is polynomial in $|S^R| $ and $|\Theta| $. The high-level idea is that, on an appropriately chosen line segment, each signal $r \in S^R $ corresponds to a jump in $p_{\mu} $, and the jump magnitude reveals the likelihoods $\{\phi^R(r|\theta)\}_{\theta} $.

The \textbf{second} challenge is that the central quantity of interest, that is $p_{\mu} $, cannot be known exactly, but can only be estimated imperfectly by repeatedly inducing posterior $\mu $ and observing the receiver's action $ a$ and the world state $\theta $. Moreover, it becomes prohibitively expensive to accurately estimate $p_{\mu}(a_0|\theta) $ near the simplex boundary, as some states have very low probability of being observed. 

The \textbf{third} challenge is that some private signals are too costly to learn in isolation. This can be due to two reasons. First, if a signal $r $ has low likelihood values, then the change in $p_{\mu} $ caused by crossing $H_r $ can be too low to detect in a reasonable number of rounds.  Second, if a signal $r $ puts too much likelihood mass on $\Theta_0 $ or $\Theta_1 $, then the corresponding hyperplane $H_r $ is close to the simplex boundary, making it hard to detect $H_r $. We introduce the notion of a $z- $balanced signal to deal with such signals:

\begin{definition}
    We call a signal $r \in S^R $ $z -$balanced if
\begin{equation}
    z \leq\frac{\sum_{\theta \in \Theta_0}\phi^R(r|\theta)}{\sum_{\theta \in \Theta}\phi^R(r|\theta)} \leq 1-z.
\end{equation}
\end{definition}
Informally, being $z- $balanced means that observing the signal $r $ still leaves the receiver uncertain about whether $\theta \in \Theta_0 $ or $\theta \in \Theta_1 $.

\textbf{Finally}, after the sender learns an imperfect estimate $\hat \phi^R $ of $\phi^R $, computing a near-optimal signaling scheme requires carefully accounting for robustness issues due to the estimation errors in the hyperplanes. Sample-efficient control over robustness requires separating learning signal likelihoods $\phi^R(r|\theta) $ from learning the hyperplane normals $n_r $.

As the sender is not able to learn the hyperplanes  exactly, we introduce the following assumption to avoid the possibility that the sender misses any of the existing polyhedral regions:

\begin{assumption}\label{as:gamma}
    There is a $\gamma_{\min} > 0 $ such that, for any cell $C $ associated with a mapping $w_C $, there exists  $\mu_C \in C $ such that, for all $r \in S^R $, we have
\begin{equation}
    |n_r^T\mu_C| \geq \gamma_{\min},
\end{equation}
where $ n_r $ is the unit normal of $H_r $. We assume that $\gamma_{\min} $ is known to the sender.
\end{assumption}

We also introduce the following assumption on the receiver's utility:

\begin{assumption}\label{as:dmin}
    The receiver's utility function $u^R $ satisfies
\begin{equation}
    d_{\min}:= \min_{\theta \in \Theta} \left | u^R(a_0, \theta) - u^R(a_1, \theta) \right | > 0.
    \end{equation}
\end{assumption}

We remark that under Assumption \ref{as:dmin}, $\Theta_0 $ and $\Theta_1$ are disjoint. We also define $d_{\max} := \max_{\theta} \left |u^R(a_0,\theta) - u^R(a_1,\theta) \right | $. Finally, we introduce the following assumption:

\begin{assumption}\label{as:distinct}
    For all $r \neq r'$, the hyperplanes $H_r$ and $H_{r'}$ are distinct.
\end{assumption}
Assumption~\ref{as:distinct} is without loss of generality. Under
Assumption \ref{as:dmin}, $H_r = H_{r'}$ precisely when the likelihood vectors
$\phi^R(r|\cdot)$ and $\phi^R(r'|\cdot)$ are proportional, in which
case the two signals induce identical posteriors in \eqref{eq: Bayes}
and can be merged into a single signal without affecting the
receiver's behavior.

Our proposed learning algorithm is presented in Algorithm \ref{alg:main}. The algorithm employs an explore-then-commit strategy. During the exploration phase, the sender constructs an estimate of $\phi^R $ by learning likelihood values $\hat \phi^R(r|\theta) $ and corresponding hyperplane normals $\hat n_r $. This is accomplished in three phases. First, the algorithm constructs a suitable line segment $l $ through the subroutine $\textsc{Construct-Search-Line} $. Then, the algorithm looks for points on $l $ where there are jumps in $p_{\mu} $, through the subroutine $ \textsc{Find-Crossings}$. Finally, the subroutine $\textsc{Learn-Signals} $ learns the unknown signal likelihoods using the output from $\textsc{Find-Crossings} $. We remark that each subroutine in Algorithm \ref{alg:main} has access to the global counter indicating the current round $t $. After the initial exploration phase, the sender computes a near-optimal signaling scheme $\hat \phi$  and commits to it for the rest of the rounds. 

Our main result is a bound on the regret suffered by Algorithm \ref{alg:main}:

\begin{theorem}\label{theorem:main}
    The regret $R_T $ achieved by Algorithm \ref{alg:main} satisfies
    \begin{equation}
          R_T \leq \tilde O \left ( \left (\frac{d_{\max}^2|S^R|^3 |\Theta|^7}{d_{\min}^2\min_{\theta}\mu_0(\theta)} \right )^{1/4}  T^{3/4} \right). 
    \end{equation}
  
\end{theorem}

Algorithm \ref{alg:main} achieves polynomial dependence on $|S^R|$ and $|\Theta|$ despite the exponentially many cells: along the randomized search line, $p_{\mu}$ is monotone and piecewise constant, reducing the problem to locating at most $|S^R|$ change points. Separately learning signal likelihoods and hyperplane normals yields logarithmic dependence on $1/\gamma_{\min}$. The $T^{3/4}$ rate combines the explore-then-commit tradeoff with the cost of estimating state-conditional action distributions near the simplex boundary.

\begin{algorithm}[t]
\caption{\textsc{Learn-to-Persuade-Binary-Receiver}}
\label{alg:main}
\begin{algorithmic}[1]
\Require  Time Horizon $T$

\State $t \gets 1 $
\State $\delta \gets 1/T  \quad $
\State Set $z_{\min}, p_{\min}, \epsilon_m, \epsilon_h   $ according to \eqref{eq:assignments}.

\Statex
\State $(\mu_-, \mu_+) \gets \textsc{Construct-Search-Line}(z_{\min}) $

\State $\eta \gets \eta_{\mathrm{sep}}(\gamma_{\min}, |\Theta|, |S^R|, z_{\min}, \delta/5) $

\Statex
\State $P \gets \Call{Find-Crossings}{\mu_-, \mu_+, \eta/4 , p_{\min}, \delta/5}$

\State $\hat \phi^R \gets \textsc{Learn-Signals}(P, \epsilon_h, \epsilon_m, p_{\min} , \delta/5) $
\State $r^0 \gets \Call{Estimate-Profile}{\mu_+, \epsilon_m, \delta/5}$
\State $r^1 \gets \mathbf{1} - \Call{Estimate-Profile}{\mu_-, \epsilon_m, \delta/5}$

\Statex
\State $\widehat \phi \gets \Call{Compute-Signaling}{\hat \phi^R ,r^0,r^1, \epsilon_h}$

\While{$t \leq T$}
    \State $\phi_t \gets \widehat \phi$
    \Comment{Commit to $\widehat \phi$ at round $t$}
    \State $t \gets t+1$
\EndWhile

\end{algorithmic}
\end{algorithm}

\subsection{$\textsc{Construct-Search-Line}$}

The algorithm starts by constructing a line segment $l $ to subsequently analyze how $p_{\mu} $ changes on $l $. This is an important step of the algorithm, as it reduces the problem of learning $\phi^R $ to identifying change-points of $p_{\mu} $ along $l $, hence helping us avoid learning each cell in isolation. We provide the pseudocode for \textsc{Construct-Search-Line} in Algorithm \ref{alg:line}. 

The construction of $l $ is simple; we uniformly sample $x_0 \sim \Delta_{\geq x_{\min}}(\Theta_0) $ and $x_1 \sim \Delta_{\geq x_{\min}}(\Theta_1) $, where $x_{\min} = \frac{1}{2|\Theta|} $, and let $l = \{(1-t)x_0 + tx_1:t \in [t_{\min}(z_{\min} ), 1-t_{\min}(z_{\min})] \} $ . Such a line segment $l $ has a number of important properties: $(i) $ for all $z_{\min} -$balanced $r \in S^R $, $l $ intersects $H_r $, $(ii)$ when restricted to $l $, $p_{\mu}(a_0|\theta) $ is monotone non-increasing and piecewise constant for all $\theta $, $(iii)$ for all $r $ such that $l \cap H_r \neq \emptyset $, $l $ intersects $H_r $ at a point that is, with high probability, away from other $H_{r'} $ for $r' \neq r $.

Since the algorithm later estimates $p_{\mu} $ along the line segment $l $, we construct $l $ such that for all $\mu \in l $, $\mu(\theta) $ is not arbitrarily small for any $\theta \in \Theta $. To this end, we sample $x_0 $ and $x_1 $ so that they put at least $x_{\min} $ probability to every state in $\Theta_0 $ and $\Theta_1 $, respectively. Similarly, the algorithm does not search over the whole segment between $x_0 $ and $x_1 $, but instead searches over a truncated version with endpoints parameterized by $t_{\min}=d_{\min}x_{\min} z_{\min}/2d_{\max} $. Although this means the algorithm may fail to catch $r \in S^R $ that are not $z_{\min}- $ balanced, we later show in Lemma \ref{lem:subopt} that the loss in utility caused by this is $O(|\Theta|z_{\min}) $.

The following lemma formalizes properties $(i)-(iii) $ by giving a high probability lower bound on the separation between $l \cap H_r $ and $H_{r'} $ for all $r,r' \in S^R $, measured within the affine hull of $\Delta(\Theta) $ :
\begin{lemma} \label{lem:separation}
    Under Assumptions \ref{as:gamma}, \ref{as:dmin} and \ref{as:distinct}, the line segment $l $ with endpoints given by $\textsc{Construct-Search-Line} $ satisfies properties $(i) $ and $(ii) $. Furthermore, for all $\delta > 0$, there is an $\eta_{\mathrm{sep}}( \gamma_{\min}, |\Theta|, |S^R|, z_{\min}, \delta) >0$ such that, with probability at least $1- \delta$, for all $H_r $ that intersects $l $, 
    \begin{equation}
        \mathrm{dist}(l \cap H_r, H_{r'}) \geq \eta_{\mathrm{sep}},  \quad \forall r' \in S^R, r' \neq r ,
    \end{equation}
    where
    \begin{equation}
        \log \left (\frac{1}{\eta_{\mathrm{sep}} } \right) = {O}\left(|\Theta|\log\left (\frac{|\Theta|d_{\max}}{d_{\min}z_{\min}\gamma_{\min}}\right )  + \log \frac{|S^R|}{ \delta}\right) .
    \end{equation}
\end{lemma}

\paragraph{Proof Sketch.}Note that for an $H_r $ that intersects the relative interior of $\Delta(\Theta) $, there exists a one-to-one mapping from $H_r \cap \Delta(\Theta)$ to $\Delta(\Theta_0) \times \Delta(\Theta_1) $.  The proof is based on the idea that, on the truncated line segment $l $, this mapping is $L(|\Theta|, z_{\min})- $Lipschitz. Thus the endpoint
pairs for which $l \cap H_r$ lands within $\eta$ of some $H_{r'}$ have
volume at most $L^{|\Theta|-2}$ times the volume of an $\eta$-strip
around $H_{r'}$ in $H_r \cap \Delta(\Theta)$. Finally, we use Assumption \ref{as:gamma} to upper bound the volume of that $\eta- $strip. 

\begin{definition}
    The event $ \mathcal{E}_{\mathrm{sep}} $ is the separation event described in Lemma \ref{lem:separation}.
\end{definition}
Lemma \ref{lem:separation} shows that, on $l $, the separation between distinct hyperplanes can be exponentially small in $|\Theta| $. Fortunately, the monotone structure of $p_{\mu} $ on $l $ allows us to use bisection when identifying the changepoints, resulting in polynomial dependence on $|\Theta| $ in the final regret bound.

\subsection{\textsc{Estimate-Profile}}

A crucial part of our algorithm is estimating $p_{\mu} $ for different $\mu \in \Delta(\Theta) $. The algorithm uses the $\textsc{Estimate-Profile} $ subroutine for this purpose, presented in Algorithm \ref{alg:query}. 

\begin{algorithm} 
\caption{$\textsc{Estimate-Profile}$}
\label{alg:query}
\begin{algorithmic}[1] 
\Require A posterior $\mu \in \mathrm{relint}( \Delta(\Theta)) $, desired error bound $\epsilon $ on $ \hat{p}_{\mu} $, failure probability $\delta $ 
\State $N \gets \lceil \frac{|\Theta|^2}{2\epsilon^2}\log \left(\frac{4|\Theta|}{\delta} \right ) \rceil $ \Comment{Number of times to observe receiver action for each state $\theta $ when the posterior is $\mu $}
\State $\phi(s_1|\theta) \gets \frac{\mu(\theta)}{\mu_0(\theta)} \quad \forall \theta \in \Theta $
\State $\phi(s_1|\theta) \gets \frac{\phi(s_1|\theta)}{\max_{\theta}\phi(s_1|\theta)} $
\State $\phi(s_2|\theta) \gets 1-\phi(s_1|\theta) \quad \forall \theta \in \Theta $
\State ${N}_{\mu}(a|\theta) \gets 0 $
\While{$\min_{\theta} \sum_{a}N_{\mu}(a|\theta) < N \text{ and } t \leq T $}
    \State $\phi_t \gets \phi $ \Comment{Commit to $\phi $}
    \State Observe $\theta \sim \mu_0(\cdot) $ and send $s \sim \phi_t(\cdot |\theta) $
    \State Observe action $a $
    \If{$s = s_1 $}
    \State    $ N_{\mu}(a|\theta)\gets N_{\mu}(a|\theta) + 1 $    
    \EndIf
    \State $t \gets t + 1 $
\EndWhile

\State $\hat p_{\mu}(\theta) \gets \frac{N_{\mu}(a_0|\theta)}{\sum_{a'} N_{\mu}(a'|\theta)} \quad \forall \theta \in \Theta$
\State \Return $\hat p_{\mu}$

\end{algorithmic}
\end{algorithm}

To induce $\mu $, the subroutine $\textsc{Estimate-Profile} $ commits to a special signaling scheme $\phi $ that induces $\mu $ with maximum probability. This signaling scheme consists of two signals, $s_1 $ and $s_2$, where $s_1 $ induces the posterior $\mu $, whereas $s_2 $ is added so that $\phi $ is a valid signaling scheme. We have the following guarantee on $\textsc{Estimate-Profile} $:
\begin{lemma}\label{lem:estimate}
    With probability at least $1-\delta $, the subroutine $\textsc{Estimate-Profile}(\mu, \epsilon, \delta) $ runs for at most
    \begin{equation}
        N(\mu, \epsilon, \delta):=\left \lceil \frac{2\sqrt{2}}{\min_{\theta}\mu(\theta) \min_{\theta}\mu_0(\theta)}\frac{|\Theta|^2}{2 \epsilon^2}\log \frac{4|\Theta|}{\delta} \right \rceil
    \end{equation}
    rounds, and returns $\hat p_{\mu} $ such that $\norm{p_{\mu} - \hat p_{\mu} }_1 \leq \epsilon $.
\end{lemma}

We remark that Algorithm \ref{alg:main} only queries posteriors $\mu $ for which 
\begin{equation}
    \min_{\theta}\mu(\theta) \geq \frac{x_{\min}^2 z_{\min} d_{\min} }{4d_{\max} }. 
\end{equation}
It follows that, with probability at least $1-\delta $, any query made by Algorithm \ref{alg:main} to \textsc{Estimate-Profile} runs for at most
\begin{equation}
    N(z_{\min}, \epsilon, \delta) := \left \lceil \frac{8\sqrt{2}d_{\max}}{x_{\min}^2 z_{\min} d_{\min} \min_{\theta}\mu_0(\theta)}\frac{|\Theta|^2}{2 \epsilon^2}\log \frac{4|\Theta|}{\delta} \right \rceil
\end{equation}
rounds and returns an $\epsilon- $accurate estimate of $p_{\mu} $.

\subsection{\textsc{Find-Crossings}}

After having generated a randomized line segment $l $, the subroutine $\textsc{Find-Crossings} $ identifies the points at which $p_{\mu} $ changes using a standard recursive binary search procedure. We use Lemma \ref{lem:separation} to determine the required binary search resolution so that we identify each crossing point on $l $ in isolation. The pseudocode for $\textsc{Find-Crossings} $ is given in Algorithm \ref{alg:crossings}. 

As some signals might have very low likelihood values, we cannot hope to identify all changes in $p_{\mu} $ along $l $ without paying a prohibitively large time cost. $\textsc{Find-Crossings} $ takes as input a minimum total signal mass $p_{\min} $ to decide how accurately to estimate $p_{\mu} $. We later show in Lemma \ref{lem:subopt} that missing signals with total likelihood mass $m(r):=\sum_{\theta}\phi^R(r \mid \theta ) $ less than $p_{\min} $ result in $O(|S^R| p_{\min}) $ loss in utility. 

We define the set $\mathcal{D} \subset S^R $ as 
\begin{equation*}
    \mathcal{D} = \{r \in S^R, m(r) \geq p_{\min}, l \cap H_r \neq \emptyset \}.
\end{equation*}

We have the following guarantee on $\textsc{Find-Crossings} $:

\begin{lemma} \label{lem:crossings}
$\textsc{Find-Crossings}(\mu_-, \mu_+, \eta_{\mathrm{sep}}/4 , p_{\min},
\delta)$ queries the subroutine $\textsc{Estimate-Profile} $ at most 
\begin{equation}
    Q_{\mathrm{cross}} := O\left ( |S^R| \log_2 (\frac{4\sqrt{2}}{\eta_{\mathrm{sep}} }) \right )
\end{equation}
times, and on the event $\mathcal{E}_{\mathrm{sep}} $, with probability at least $1- \delta $, 
runs for at most
\begin{equation}
      Q_{\mathrm{cross}} N(z_{\min}, p_{\min}/6 , \delta/Q_{\mathrm{cross} } )
\end{equation}
rounds, and outputs
$\{(u_i, v_i)\}_{i=1}^{k}$ that satisfies:
(i)~for every $r \in \mathcal{D}$ there is an $i \leq k$ with
$n_r^\top u_i \geq 0 \geq n_r^\top v_i$;
(ii)~for every $i \leq k$ there is $r(i) \in S^R$ with
$n_{r(i)}^\top u_i \geq 0 \geq n_{r(i)}^\top v_i$, $m(r(i)) \geq p_{\min}/6 $ and $a^*(u_i, r(i))\neq a^*(v_i, r(i)) $;
(iii)~$\norm{u_i - v_i}_2 \leq \eta_{\mathrm{sep}}/4 $ for all $i \leq k$.
\end{lemma}

\begin{definition}
    On $\mathcal{E}_{\mathrm{sep}} $, $\mathcal{E}_{\mathrm{cross}} $ is the event that the output of $\textsc{Find-Crossings}(\mu_-, \mu_{+}, \eta_{\mathrm{sep}}/4 , p_{\min}, \delta) $ satisfies the conditions $(i)-(iii) $ in Lemma \ref{lem:crossings}.
\end{definition}

Condition $(i) $ in Lemma \ref{lem:crossings} means that, for every signal $r \in D $ the algorithm aims to detect, $\textsc{Find-Crossings} $ returns a pair of posteriors on the opposite sides of $H_r $. Conversely, $(ii) $ means that, for every pair that is returned by $\textsc{Find-Crossings} $, there is indeed a hyperplane $H_r $ between them with total mass at least $p_{\min}/6 $. We remark that, on $\mathcal{E}_{\mathrm{sep}} $, the pair $(u_i, v_i) $ contains a single hyperplane between them. 

\subsection{\textsc{Learn-Signals}}

\begin{algorithm}[t]
    \caption{\textsc{Learn-Signals} }
    \begin{algorithmic}
        \Require  Set $P $ of pairs of posteriors,  angular hyperplane accuracy $\epsilon_h $, likelihood accuracy $\epsilon_m $,  failure probability $\delta $
        \State $\hat \phi^R \gets \emptyset $, $k \gets |P| $
        \For{$(u, v) \in P $}
            \State $\hat p \gets \textsc{Estimate-Likelihood}(u, v, \epsilon_m, \delta/2k )$
            \State $\hat n \gets \textsc{Learn-Hyperplane}(u ,v, \epsilon_h, \delta/2k) $
            \State $\hat \phi^R \gets \hat \phi^R \cup \{(\hat n, \hat p)\} $
        \EndFor
        \State \Return $\hat \phi^R $
    \end{algorithmic}
\end{algorithm}

Once $\textsc{Find-Crossings}$ has returned a pair of posteriors $u_r, v_r$
whose connecting segment contains a single, isolated crossing point of
$H_r$, the algorithm invokes the $\textsc{Learn-Signals}$ subroutine to
estimate the signal likelihoods $\{\phi^R(r|\theta)\}_{\theta \in \Theta}$
and the corresponding hyperplane unit normal $n_r$. We provide the pseudocode for the subroutines \textsc{estimate-likelihood} and \textsc{learn-hyperplane} in Algorithm \ref{alg:likelihood} and Algorithm \ref{alg:hyperplane}, respectively. The following lemma formalizes the sample complexity of $\textsc{Learn-Signals} $:

\begin{lemma}\label{lem:learn-signal}
    Conditioned on $\mathcal{E}_{\mathrm{sep}} $ and $\mathcal{E }_{\mathrm{cross}} $, let $P =\{(u_i, v_i)\}_{i=1}^k $ be the set of intervals returned by $\textsc{Find-Crossings} $, and let $r(i) \in S^R $ be the corresponding signal. $\textsc{Learn-Signals}(P, \epsilon_h, \epsilon_m, p_{\min} ,\delta) $ queries $\textsc{Estimate-Profile} $ at most 
    \begin{equation}
        Q_{\mathrm{learn}}:= O \left (|\Theta|^2|S^R| + |\Theta||S^R| \log_2 (\frac{32 |\Theta|}{\epsilon_h})  \right )
    \end{equation}
    times. Furthermore, with probability at least $1- \delta $, \textsc{Learn-Signals}($P, \epsilon_h, \epsilon_m ,p_{\min}, \delta$) runs for at most \begin{equation}
        Q_{\mathrm{learn}} N (z_{\min}, \min\{\epsilon_m, p_{\min}/48\} , \delta / Q_{\mathrm{learn}})
    \end{equation}
    rounds and returns $\hat \phi^R = \{\hat n_i, \hat p_i\}_{i=1}^k $  such that, for all $i \in [k] $, $\norm{\hat p_i - \phi^R_{r(i)} }_1 \leq \epsilon_m$ and $ \norm{\hat n_i-n_{r(i)} }_2 \leq \epsilon_h $.
\end{lemma}

\begin{definition} On $\mathcal{E}_{\mathrm{sep}} $ and $\mathcal{E}_{\mathrm{cross}} $,  $\mathcal{E}_{\mathrm{learn}} $ is the event described in Lemma \ref{lem:learn-signal}.
\end{definition}

A crucial design choice in $\textsc{Learn-Signals}$ is to decouple the
estimation of the likelihood values from the estimation of the hyperplane.
At first glance this separation may seem redundant: $n_r$ is fully
determined by the likelihoods $\{\phi^R(r|\theta)\}_{\theta}$ together with
the (known) receiver utilities, so a naive approach would estimate the
likelihoods directly from the difference $p_{u_r} - p_{v_r}$ and construct
an estimate $\hat n_r$ of $n_r$ from them. The issue is one of error
amplification: while a coarse likelihood estimate suffices for evaluating
the sender's utility, guaranteeing that every cell of the partition is
classified correctly requires angular accuracy
$\norm{n_r - \hat n_r}_2 \leq c\,\gamma_{\min}$ (cf. Assumption \ref{as:gamma}).
Achieving this accuracy from the difference $p_{u_r} - p_{v_r}$ alone
would require a number of exploration rounds scaling as
$O(1/\gamma_{\min}^2)$, which then propagates to the regret. $\textsc{Learn-Signals}$ instead reduces the problem of learning $n_r$ to
actively learning a halfspace: the sign of $n_r^\top \mu$ determines on
which side of $H_r$ a posterior $\mu$ lies, and this sign can be queried
by inducing $\mu$ and comparing the estimate returned by
$\textsc{Estimate-Profile}$ against the profiles at $u_r$ and $v_r$. Since
each membership query localizes $H_r$ up to a constant factor,
$\textsc{Learn-Signals}$ attains angular accuracy $\gamma_{\min}$ using
only $O(\log(1/\gamma_{\min}))$ queries, exponentially improving over the
naive approach. 

\subsection{\textsc{Compute-Signaling}}

After having learned an estimate $\hat \phi^R $ of $\phi^R $, the sender can compute a near-optimal signaling scheme $\phi $ for the remaining rounds. We denote by $\hat S^R $ the signal alphabet of $\hat \phi^R $.

Before computing a signaling scheme, the algorithm introduces two special receiver signals, called $r^0 $ and $r^1 $. These are not signals in the usual sense, but instead always cause the receiver to take actions $a_0 $ and $a_1 $, respectively. Since the algorithm is only guaranteed to detect $z_{\min} -$balanced signals, $r^0 $ ($r^1$) accounts for the \emph{aggregate} effect of signals in $S^R $ that put more than $1-z_{\min} $ fraction of their total mass on $\Theta_0 $ ($\Theta_1 $), and whose hyperplanes are not intersected by the search line $l $. The likelihood values of $r^0 $ and $r^1 $ are estimated directly from the receiver's behavior at the two ends of the search line: at posterior $\mu_+ $, all signals that put less than $z_{\min} $ fraction of their total mass on $\Theta_1 $ and not intersected by $l $ induce $a_0 $. Hence the state-conditional probability of observing $a_0 $ at $\mu_+ $ reveals the aggregate likelihood of such signals. The likelihood values of $r^1 $ are computed similarly. We define the events $\mathcal{E}_0 $ and $\mathcal{E}_1 $ as the events that aggregate likelihoods of $r^0 $ and $r^1 $, respectively, are both estimated to accuracy $\epsilon_m $.

$\textsc{Compute-Signaling}$ solves the program in \eqref{eq:persuasion} with $\hat \phi^R $ in place of $\phi^R $, with the receiver's estimated signal alphabet being $\hat S^R \cup \{r^0, r^1\} $,  and with the following modifications to the feasible set:
\begin{subequations}\label{eq:compute}
\begin{align}
    &y_{s,r^0} = 0, y_{s,r^1} = 1, \quad  s \in S, \\
    &(-1)^{y_{s,r }}\sum_{\theta\in\Theta}
  \mu_s(\theta)\hat n_r(\theta) \geq 2 \epsilon_h ,\quad (s,r)\in S\times \hat  S^R ,
\label{eq:compute-threshold} \\
    & \sum_{\theta\in\Theta}
  \mu_0(\theta)\phi_\theta(s^b)
  \le
  \frac{4\sqrt{2}\epsilon_h}
       {\gamma_{\min}\min_{\theta\in\Theta}\mu_0(\theta)} .
\label{eq:compute-bound}
\end{align}   
\end{subequations}
We provide the pseudocode for $\textsc{Compute-Signaling} $ in Algorithm \ref{alg:commit}. Comparing the program \eqref{eq:compute} with the original program \eqref{eq:persuasion}, we first observe that $\textsc{Compute-Signaling} $ uses a pessimistic version of the incentive-compatibility constraint in \eqref{eq:incentive} by introducing a margin in \eqref{eq:compute-threshold}. In other words, the signaling scheme computed by the sender is allowed to induce only posteriors that are sufficiently distant from the estimated hyperplanes. Without this margin, a posterior $\mu_s $ induced by the sender can fall on the wrong side of a true hyperplane, resulting in large regret. Note that we do not impose such a margin constraint for $r^0 $ or $r^1 $, as they do not correspond to a real hyperplane in \eqref{eq:persuasion}. Finally, we allow the sender to use an additional, unconstrained, low-probability signal $s^b $, meaning that the sender uses $|\Theta|+1 $ signals. 

We have the following guarantee on the signaling scheme $\hat \phi $ returned by  $\textsc{Compute-Signaling} $:

\begin{lemma}\label{lem:subopt}
    On the events $\mathcal{E}_{\mathrm{sep}}, \mathcal{E}_{\mathrm{cross}}, \mathcal{E}_{\mathrm{learn}},\mathcal{E}_0, \mathcal{E}_1 $, the signaling scheme $\hat \phi $ returned by $\textsc{Compute-Signaling}(\hat \phi^R, r^0, r^1, \epsilon_h) $ is \begin{equation}\label{eq:subopt}
        O \left (|\Theta|z_{\min} + |S^R|p_{\min} + |S^R|\epsilon_m + \frac{\epsilon_h}{\gamma_{\min}\min_{\theta}\mu_0(\theta)} \right )-
    \end{equation}
    optimal.
\end{lemma}

The four terms in Lemma~\ref{lem:subopt} correspond to: merging
signals that are not $z_{\min}$-balanced into $r^0$ and $r^1$;
missing signals of total mass below $p_{\min}$; likelihood error
$\epsilon_m$; and the cost of the robustness margin in
\eqref{eq:compute-threshold}. 

\paragraph{Proof sketch.}

Similar to standard persuasion, we can define the value of a posterior $\mu $ to the sender in the problem of persuading a privately informed receiver as $V(\mu) = \sum_{\theta,a}\mu(\theta)p_{\mu}(a|\theta)u_{\theta}(a) $. Note that, after having learned $\hat \phi^R $, the sender can use it to compute an estimate $\hat p_{\mu} $ of $p_{\mu} $ for all $\mu \in \Delta(\Theta) $. The main idea behind the proof is to compare the value of a posterior $\mu $ in the original problem \eqref{eq:persuasion} and the sender's problem \eqref{eq:compute}, where the difference is controlled by $\norm{p_{\mu} - \hat p_{\mu}}_1 $. Whenever $\mu $ satisfies the margin constraint \eqref{eq:compute-threshold}, the same mapping $w:\hat S^R \rightarrow A $ is associated with $\mu $ in the original and the estimated problem, in which case the error is controlled by the first three terms in \eqref{eq:subopt}. To bound the utility lost by imposing the margin, we construct a feasible scheme for \eqref{eq:compute} that is near-optimal for \eqref{eq:persuasion}. Similar to \citet{li2025information}, we tilt each posterior induced by an optimal scheme away from the hyperplanes by taking a convex combination with a $\gamma_{\min} -$deep point of its corresponding cell ($\mu_C $ in Assumption \ref{as:gamma}), and restore Bayes-plausibility with one additional low-probability signal. This is the reason we allow the sender to use an additional low-probability signal $s^b $ in \eqref{eq:compute-bound}.

\paragraph{Remark.}$\textsc{Compute-Signaling}$ runs in time
$O(2^{|\Theta||S^R|}\operatorname{poly}(|\Theta|))$. This exponential
dependence is a computational limitation of our algorithm, although we
emphasize that it does not affect the regret guarantee, which is
information-theoretic.
Moreover, we remark that \citet{hossain2024multi} show that computing a sender's best response in multi-sender persuasion is NP-hard. However, their hardness result does not apply to the binary-action receiver but requires large action spaces. Whether the offline problem admits a polynomial-time algorithm for binary-action receivers remains open to the best of our knowledge. 

\paragraph{Putting it together.}
We now sketch the proof of Theorem~\ref{theorem:main}. Setting
$\delta = 1/T$, the five events underlying
Lemmas~\ref{lem:separation}--\ref{lem:subopt} jointly fail with
probability at most $\delta$, contributing $O(1)$ regret. On their
intersection, exploration lasts at most
$(Q_{\mathrm{cross}} + Q_{\mathrm{learn}} + 2) \cdot
N(z_{\min}, p_{\min}/48, \delta/5)$ rounds, each
incurring $O(1)$ regret, while every commitment round incurs the
suboptimality of Lemma~\ref{lem:subopt}. The exploration cost scales
as $1/(z_{\min}\epsilon_m^2)$ and the commitment loss as
$(z_{\min} + \epsilon_m)$ per round; balancing the two yields the
choices $z_{\min}, p_{\min}, \epsilon_m, \epsilon_h = \Theta(T^{-1/4})$
in Algorithm~\ref{alg:main} and the claimed bound.
\hfill$\square$

\section{Conclusion}

We studied online Bayesian persuasion with a privately informed binary-action receiver under action feedback. We proposed a no-regret algorithm with regret scaling polynomially with the sizes of the state space and the receiver's signal alphabet. 

Our work opens several directions for future research. First, we do not
know whether the $T^{3/4}$ rate is tight: it might be possible to achieve better rates through adaptive exploration. Another natural extension of our work is to study receivers with more than two actions. 

\section*{Acknowledgments}

This work was supported in part by ONR N00014-24-1-2432 and ARO W911NF-23-1-0317.

\bibliography{ref}

\newpage
\onecolumn
\appendix

\counterwithin{equation}{section}
\counterwithin{theorem}{section}
\counterwithin{algorithm}{section}
\setcounter{theorem}{0}
\renewcommand{\thetheorem}{\thesection\arabic{theorem}}

\setcounter{equation}{0}
\renewcommand{\theequation}{\thesection\arabic{equation}}

\setcounter{definition}{0}
\renewcommand{\thedefinition}{\thesection\arabic{definition}}

\setcounter{assumption}{0}
\renewcommand{\theassumption}{\thesection\arabic{assumption}}

\setcounter{algorithm}{0}
\renewcommand{\thealgorithm}{\thesection\arabic{algorithm}}
\setcounter{secnumdepth}{1}

    \section{Proof of Theorem \ref{theorem:main}}

    \subsection{Parameter Assignments for Algorithm \ref{alg:main}}

    We first provide the assignments for the parameters $z_{\min}, p_{\min}, \epsilon_m, \epsilon_h $ used by Algorithm \ref{alg:main}. 

            Let 
        \begin{equation}
            n:= |\Theta|, \qquad m = |S^R|, \qquad \mu_{\min} = \min_{\theta}\mu_0(\theta), \qquad \kappa = \frac{d_{\max}}{d_{\min}},
        \end{equation}
        and define
        \begin{equation}
            L_T := \log \left (e + \frac{Tmn\kappa(1+\kappa)}{\gamma_{\min}\mu_{\min}} \right ),
        \end{equation}
        \begin{equation}\label{eq:rho-unclipped}
    \rho_*
    :=
    \left(
        \frac{\kappa(1+\kappa)m^3n^6L_T(n+L_T)}
             {\mu_{\min}T}
    \right)^{1/4},
    \qquad
    \rho:=\min\{1,\rho_*\}.
        \end{equation}

        Algorithm \ref{alg:main} uses 
        \begin{equation}\label{eq:assignments}
            \delta = \frac{1}{T}, \qquad z_{\min} = \frac{\rho}{32n(1 + \kappa)}, \qquad p_{\min}, \epsilon_m = \frac{\rho}{16m}, \qquad \epsilon_h = \frac{\gamma_{\min}\mu_{\min}\rho}{64(1 + 2\sqrt{2})}.
        \end{equation}

        The clipping in \eqref{eq:rho-unclipped} guarantees that all four
accuracy parameters lie in their required ranges. In particular,
\(z_{\min}\leq 1/(32n(1+\kappa))\),
\(p_{\min}=\epsilon_m\leq1/(16m)\), and
\(\epsilon_h\leq\gamma_{\min}/[8(1+2\sqrt2)]\).

\subsection{Proof of Theorem \ref{theorem:main}}

    \begin{proof}
    We first state the explicit rate achieved by Algorithm \ref{alg:main}:
\begin{align}\label{eq:explicit}
    R_T 
    &= O \left( \min \left \{ T,  \left(
        \frac{\kappa(1+\kappa)m^3n^6L_T(n+L_T)}
             {\mu_{\min}}
    \right)^{1/4}T^{3/4} \right \} \right )
\end{align}
    
Let $\mathcal{G} = \mathcal{E}_{\mathrm{sep}}\cap \mathcal{E}_{\mathrm{cross}}\cap \mathcal{E}_{\mathrm{learn}} \cap \mathcal{E}_{0} \cap \mathcal{E}_{\mathrm{1}} $ .The first event fails with probability at most $\delta/5 $, the next two conditional events fail with at most probability $\delta/5 $ under the preceding events, and the last two events each fail with probability at most $\delta/5 $. Hence we have
\begin{equation}\label{eq:good-event-probability}
    \mathbb{P}(\mathcal{G}^c) \leq \delta.
\end{equation}

If $\rho = 1 $, this means that the second term in \eqref{eq:explicit} is greater than $T $. Since utilities are bounded in $[0,1] $, in this regime, we can simply bound the regret suffered by the algorithm as
\begin{equation}
    R_T \leq T.
\end{equation}

For the rest of the proof, assume that
\begin{equation}\label{eq:nontrivial-regime}
    \rho = \rho^* < 1.
\end{equation}

Since $\kappa(1+\kappa)m^3n^6L_T(n+L_T)/\mu_{\min} \geq 1 $, \eqref{eq:rho-unclipped} implies
\begin{equation}\label{eq:rho-lower-bound}
    \rho \geq T^{-1/4}, \qquad \log \frac{1}{\rho} \leq \frac{1}{4}\log T.
\end{equation}

We first bound the duration of exploration $T_{\mathrm{exp}} $ on the event $\mathcal{G} $. Let
\(Q_{\mathrm{cross}}\) and \(Q_{\mathrm{learn}}\) be the deterministic
query bounds in Lemmas~\ref{lem:crossings} and
\ref{lem:learn-signal}, and put
\begin{equation}
    Q := Q_{\mathrm{cross}} + Q_{\mathrm{learn}} + 2.
\end{equation}
Lemmas \ref{lem:crossings} and \ref{lem:learn-signal}, together with the two final calls to $\textsc{Estimate-Profile} $ for learning $r^0 $ and $r^1 $ yields
\begin{align}
    T_{\mathrm{exp}}
&\leq
 Q_{\mathrm{cross}}
 N\!\left(
    z_{\min},\frac{p_{\min}}6,
    \frac{\delta}{5Q_{\mathrm{cross}}}
 \right)
\notag\\
&\quad+
 Q_{\mathrm{learn}}
 N\!\left(
    z_{\min},
    \min\!\left\{\epsilon_m,\frac{p_{\min}}{48}\right\},
    \frac{\delta}{5Q_{\mathrm{learn}}}
 \right)
 +2N\!\left(z_{\min},\epsilon_m,\frac{\delta}{5}\right).
\label{eq:Texp-sum}
\end{align}
The function $N $ is non-increasing in its accuracy and failure probability arguments. By \eqref{eq:assignments},
\begin{equation}\label{eq:bar-epsilon}
        \bar\epsilon
    :=\min\!\left\{
       \frac{p_{\min}}6,
       \epsilon_m,
       \frac{p_{\min}}{48}
    \right\}
    =\frac{\rho}{768m}.
\end{equation}
Since $Q > Q_{\mathrm{cross}}, Q_{\mathrm{learn}}, 2 $, \eqref{eq:Texp-sum} implies
\begin{equation}\label{eq:Texp-compressed}
        T_{\mathrm{exp}}
    \leq
    QN\!\left(
       z_{\min},\bar\epsilon,\frac{\delta}{5Q}
    \right).
\end{equation}

We next control the logarithmic factors. Lemma~\ref{lem:separation},
used with failure probability \(\delta/5\), gives
\begin{equation}\label{eq:eta-rate-theorem}
 \log\frac1{\eta_{\mathrm{sep}}}
 =O\!\left(
      n\log\!\frac{n\kappa}
                       {z_{\min}\gamma_{\min}}
      +\log\frac{5m}{\delta}
    \right).
\end{equation}
Using \eqref{eq:assignments} and \eqref{eq:rho-lower-bound}, and noting that $\kappa \geq 1 $, $\mu_{\min} \leq 1 $ and $\gamma_{\min} \leq 1 $, we get
\begin{equation}
    \log\!\frac{n\kappa}
                       {z_{\min}\gamma_{\min}} = O(L_T), \qquad \log \frac{5m}{\delta} = O(L_T), \qquad \log\frac{32n}{\epsilon_h} = O(L_T).
\end{equation}
Consequently, Lemmas~\ref{lem:crossings} and
\ref{lem:learn-signal} imply
\begin{align}
Q_{\mathrm{cross}}
&=O(mnL_T),
\label{eq:Qcross-theorem}\\
Q_{\mathrm{learn}}
&=O(mn^2+mnL_T)
 =O\!\bigl(mn(n+L_T)\bigr).
\label{eq:Qlearn-theorem}
\end{align}
Therefore
\begin{equation}\label{eq:Q-total-theorem}
    Q=O\!\bigl(mn(n+L_T)\bigr).
\end{equation}
In particular,
\begin{equation}\label{eq:N-log-theorem}
    \log\!\left(\frac{20nQ}{\delta}\right)=O(L_T).
\end{equation}

Recall that \(x_{\min}=1/(2n)\). The definition of
\(N(z_{\min},\epsilon,\delta')\) in Lemma~\ref{lem:estimate}, together
with \eqref{eq:assignments}, \eqref{eq:bar-epsilon}, and
\eqref{eq:N-log-theorem}, gives
\begin{align}
N\!\left(
   z_{\min},\bar\epsilon,\frac{\delta}{5Q}
 \right)
&\leq
1+
\frac{4\sqrt2\kappa}
     {x_{\min}^2z_{\min}\mu_{\min}}
\frac{n^2}{2\bar\epsilon^2}
\log\!\left(\frac{20nQ}{\delta}\right)
\notag\\
&=O\!\left(
   \frac{\kappa(1+\kappa)m^2n^5}
        {\mu_{\min}\rho^3}
   L_T
 \right).
\label{eq:N-theorem}
\end{align}

Combining
\eqref{eq:Texp-compressed}, \eqref{eq:Q-total-theorem}, and
\eqref{eq:N-theorem}, we obtain
\begin{equation}\label{eq:Texp-final}
    T_{\mathrm{exp}}
    =O\!\left(
       \frac{\kappa(1+\kappa)m^3n^6}
            {\mu_{\min}\rho^3}
       L_T(n+L_T)
    \right).
\end{equation}
By \eqref{eq:rho-unclipped} and \eqref{eq:nontrivial-regime},
\begin{equation}\label{eq:balance-exploration}
    \frac{\kappa(1+\kappa)m^3n^6L_T(n+L_T)}
         {\mu_{\min}\rho^3}
    =T\rho.
\end{equation}
Thus \(T_{\mathrm{exp}}=O(T\rho)\) on \(\mathcal G\).

It remains to bound the commitment loss. The explicit conclusion of
Lemma~\ref{lem:subopt} gives, on \(\mathcal G\),
\begin{align}
V^*-V(\widehat\phi)
&\leq
 \left(
    \frac{4\sqrt2}{\gamma_{\min}\mu_{\min}}
    +\frac{16}{\gamma_{\min}}
 \right)\epsilon_h
 +8(1+\kappa)nz_{\min}
 +4mp_{\min}+4m\epsilon_m
\notag\\
&\leq
 \frac{16(1+2\sqrt2)}
      {\gamma_{\min}\mu_{\min}}\epsilon_h
 +8(1+\kappa)nz_{\min}
 +4mp_{\min}+4m\epsilon_m
\notag\\
&=\frac{\rho}{4}+\frac{\rho}{4}
  +\frac{\rho}{4}+\frac{\rho}{4}
 =\rho,
\label{eq:commitment-gap}
\end{align}
where the second line uses \(\mu_{\min}\leq1\), and the last line uses
\eqref{eq:assignments}. Hence the commitment phase contributes at
most \(T\rho\) regret.

Every exploration round and every round on \(\mathcal G^c\) incurs
regret at most one. Using \eqref{eq:good-event-probability},
\eqref{eq:Texp-final}, \eqref{eq:balance-exploration}, and
\eqref{eq:commitment-gap}, we conclude that
\begin{align}
R_T
&\leq O(T\rho)+T\rho+\delta T
 =O(T\rho)+1
 =O(T\rho)
\notag\\
&=O\!\left(
    \left(
        \frac{\kappa(1+\kappa)}{\mu_{\min}}
    \right)^{1/4}
    m^{3/4}n^{3/2}T^{3/4}
    \bigl[L_T(n+L_T)\bigr]^{1/4}
  \right),
\label{eq:optimized-regret}
\end{align}
where \(1\leq T\rho\) follows from \eqref{eq:rho-lower-bound}.
Together with the trivial regime \(\rho=1\), this concludes the proof.
\end{proof}

\newpage

\section{Proof of Lemma \ref{lem:separation}}

\begin{algorithm}
\caption{\textsc{construct-search-line}}
\label{alg:line}
\begin{algorithmic}[1]
    \Require $z_{\min}$
    \State $x_{\min} \gets \dfrac{1}{2|\Theta|}$
    \Statex
    \State $t_{\min}
        \gets
        \dfrac{d_{\min}x_{\min}z_{\min}}{2d_{\max}}$
    \Statex
    \State Sample
        $x_0
        \sim
        \mathrm{Uniform}
        \bigl(
            \Delta_{\geq x_{\min}}(\Theta_0)
        \bigr)$
    \State Sample
        $x_1
        \sim
        \mathrm{Uniform}
        \bigl(
            \Delta_{\geq x_{\min}}(\Theta_1)
        \bigr)$
    \State $\mu_-
        \gets
        t_{\min}x_1
        +
        (1-t_{\min})x_0$
    \State $\mu_+
        \gets
        (1-t_{\min})x_1
        +
        t_{\min}x_0$
    \State \Return $\mu_-,\mu_+$
\end{algorithmic}
\end{algorithm}

We first prove that $l$ satisfies properties $(i)$ and $(ii)$ in Lemma
\ref{lem:separation}, which is relatively straightforward.

\begin{proof}
    Let
    \[
        d_{\theta}
        =
        u^R(a_0,\theta)
        -
        u^R(a_1,\theta).
    \]

    \paragraph{Property $(i)$:}
    Let $r$ be a $z_{\min}$-balanced signal. The point
    \[
        (1-t_r)x_0+t_rx_1,
    \]
    with $t_r$ given as
    \begin{align}
        t_r
        &=
        \frac{
            \displaystyle
            \sum_{\theta\in\Theta_0}
            x_0(\theta)
            \phi^R(r\mid\theta)
            |d_{\theta}|
        }{
            \displaystyle
            \sum_{\theta\in\Theta_0}
            x_0(\theta)
            \phi^R(r\mid\theta)
            |d_{\theta}|
            +
            \sum_{\theta\in\Theta_1}
            x_1(\theta)
            \phi^R(r\mid\theta)
            |d_{\theta}|
        }
        \notag\\
        &\geq
        \frac{
            x_{\min}d_{\min}
            \displaystyle
            \sum_{\theta\in\Theta_0}
            \phi^R(r\mid\theta)
        }{
            d_{\max}
            \displaystyle
            \sum_{\theta}
            \phi^R(r\mid\theta)
        }
        \notag\\
        &\geq
        \frac{
            d_{\min}x_{\min}z_{\min}
        }{
            d_{\max}
        }
        \notag\\
        &\geq
        t_{\min},
    \end{align}
    lies on $H_r$. Similarly we have $1-t_r\geq t_{\min}$. It follows that
    $l\cap H_r\neq\emptyset$.

    \paragraph{Property $(ii)$.}
    Let
    \[
        \mu_t
        =
        tx_1+(1-t)x_0.
    \]
    The preferred action of signal $r$ at $\mu_t$, that is
    $a^*(\mu_t,r)$ depends on $\operatorname{sign}(n_r^\top\mu_t)$.
    Since we have
    \[
        p_{\mu}(\theta)
        =
        \sum_{r:a^*(\mu,r)=a_0}
        \phi^R(r\mid\theta),
    \]
    it follows that $p_{\mu_t}$ is piecewise constant for all
    $\theta\in\Theta$. To see the monotonicity of $p_{\mu_t}$, observe
    that for $t<t_r$, $a^*(\mu_t,r)=a_0$ whereas for $t>t_r$,
    $a^*(\mu_t,r)=a_1$.
\end{proof}

It remains to establish the separation property~$(iii)$. We first
handle the case \( |\Theta|=2 \). In this case,
\(\operatorname{aff}(\Delta(\Theta))\) is one-dimensional, and the
intersection of each \(H_r\) with
\(\operatorname{aff}(\Delta(\Theta))\) is a single point. Consider two
distinct hyperplanes \(H_r\) and \(H_{r'}\), and let \(C\) be a cell
lying between their intersection points. By
Assumption~\ref{as:gamma}, there exists \(\mu_C\in C\) such that
\[
    |n_r^\top\mu_C|
    \geq
    \gamma_{\min}
    \qquad\text{and}\qquad
    |n_{r'}^\top\mu_C|
    \geq
    \gamma_{\min}.
\]
Let
\[
    T
    =
    \left\{
        z\in\mathbb{R}^{|\Theta|}
        :
        \boldsymbol{1}^\top z=0
    \right\}
\]
be the tangent space of the simplex. Since
\begin{align}
    \operatorname{dist}_{\operatorname{aff}(\Delta(\Theta))}
    (\mu_C,H_s)
    &=
    \frac{
        |n_s^\top\mu_C|
    }{
        \|P_Tn_s\|_2
    }
    \notag\\
    &\geq
    |n_s^\top\mu_C|,
    \qquad
    s\in\{r,r'\},
\end{align}
the distance between the two intersection points is at least
\(2\gamma_{\min}\). Consequently,
\[
    \operatorname{dist}(l\cap H_r,H_{r'})
    \geq
    \gamma_{\min}
\]
whenever \(H_r\) intersects \(l\). Thus property~$(iii)$ holds
deterministically for \( |\Theta|=2 \) with
\[
    \eta_{\mathrm{sep}}
    =
    \gamma_{\min}.
\]
For the remainder of the proof, suppose that \( |\Theta|\geq3 \).

We now prove the following pairwise separation result for two distinct
hyperplanes \(H_r\) and \(H_{r'}\). Let $T = \{z :\boldsymbol{1}^\top z = 0 \} $ be the simplex tangent space, and let $P_{T} $ be the orthogonal projector onto $T $. The angle $\alpha_{r,r'} \in [0, \pi/2] $ between distinct hyperplanes $H_r $ and $H_{r'} $, measured within simplex affine hull, is defined as the angle between $P_Tn_r $ and $P_Tn_{r'} $.

\begin{lemma}
\label{lem:pairwise}
    Assume $|\Theta|\geq3$. Let $H_r$ and $H_{r'}$ be two distinct
    hyperplanes separating $\Delta(\Theta_0)$ from
    $\Delta(\Theta_1)$. Suppose we draw
    \[
        x_0
        \sim
        \mathrm{Uniform}
        \bigl(
            \Delta_{\geq x_{\min}}(\Theta_0)
        \bigr)
    \]
    and
    \[
        x_1
        \sim
        \mathrm{Uniform}
        \bigl(
            \Delta_{\geq x_{\min}}(\Theta_1)
        \bigr).
    \]
    Let
    \[
        d_{r,r'}
        :=
        \operatorname{dist}
        \bigl(
            H_r\cap\Delta(\Theta),
            H_{r'}
        \bigr)
    \]
    and $\alpha_{r,r'}$ be the angle between $H_r$ and $H_{r'}$,
    measured within the simplex affine hull.

    Consider the line segment defined as
    \[
        l
        =
        \left\{
            (1-t)x_0+tx_1
            :
            t\in[t_{\min},1-t_{\min}]
        \right\},
    \]
    and assume $l\cap H_r\neq\emptyset$. Then, with probability at
    least $1-\delta$, for some constant $C_{|\Theta_0|, |\Theta_1|} $, we have
    \begin{equation}
        \operatorname{dist}(l\cap H_r,H_{r'})
        \geq
        \max
        \left\{
            \delta
            \sin(\alpha_{r,r'})
            C_{|\Theta_0|,|\Theta_1|}
            \left(
                \frac{2}{t_{\min}^2}
            \right)^{1-\frac{|\Theta|}{2}},
            \,
            d_{r,r'}
        \right\}.
    \end{equation}
\end{lemma}

\begin{proof}
    Let
    \[
        A
        :=
        \Delta_{\geq x_{\min}}(\Theta_0)
        \qquad\text{and}\qquad
        B
        :=
        \Delta_{\geq x_{\min}}(\Theta_1),
    \]
    and let
    \[
        D
        =
        \left\{
            tp+(1-t)q
            :
            p\in A,\,
            q\in B,\,
            t\in[t_{\min} ,1-t_{\min}]
        \right\}
    \]
    be the \emph{trimmed} simplex. Consider, for $r $ such that $l \cap H_r \neq \emptyset $, the mapping from $A\times
    B$ to $H_r\cap D$, defined as
    \[
        \Phi(p,q)
        =
        l(p,q)\cap H_r,
    \]
    where we defined 
    \begin{equation}
        l(p,q) := \{tp + (1-t)q, t \in [0,1]\}.
    \end{equation}
    We first remark that $\Phi$ is one-to-one. To see that, let
    $(p_1,q_1)\in A\times B$ and $(p_2,q_2)\in A\times B$, and assume
    that $\exists t_1,t_2\in(0,1)$, we have
    \[
        t_1p_1+(1-t_1)q_1
        =
        t_2p_2+(1-t_2)q_2.
    \]
    It follows that
    \[
        p_1
        =
        \frac{t_2}{t_1}p_2,
    \]
    which can only happen if $t_2=t_1$ and $p_2=p_1$, as
    $p_1,p_2\in\Delta(\Theta_0)$. Same argument goes for $q_1$ and
    $q_2$. Hence we have
    \[
        (p_1,q_1)
        =
        (p_2,q_2).
    \]
    To see that for all $x\in H_r\cap D$, we have an associated
    $(p,q)\in A\times B$, simply let
    \[
        p_x(\theta)
        =
        \frac{
            x(\theta)
        }{
            \displaystyle
            \sum_{\theta'\in\Theta_0}
            x(\theta')
        } \geq x_{\min}
        \quad
        \text{for }\theta\in\Theta_0
    \]
    and
    \[
        q_x(\theta)
        =
        \frac{
            x(\theta)
        }{
            \displaystyle
            \sum_{\theta'\in\Theta_1}
            x(\theta')
        }\geq x_{\min}
        \quad
        \text{for }\theta\in\Theta_1,
    \]
    and observe that, with
    \[
        t_x
        =
        \sum_{\theta'\in\Theta_0}
        x(\theta'),
    \]
    \[
        x
        =
        t_xp_x+(1-t_x)q_x.
    \]
    It follows that the inverse map
    \[
        \Phi^{-1}
        :
        H_r\cap D
        \longrightarrow
        A\times B
    \]
    exists. To proceed with the proof, we will first show that
    $\Phi^{-1}$ is Lipschitz, then use the Lipschitzness to show that a
    bad set (a set that is close to $H_{r'}$) on $H_r\cap D$ has small
    pre-image on $A\times B$.

    \textbf{Step 1: Lipschitzness of $\Phi^{-1}$.}
    To show the Lipchitzness of $\Phi^{-1}$, consider a perturbation
    $v$ such that
    \[
        \sum_{\theta}v(\theta)
        =
        0.
    \]
    Note that we have
    \[
        t_{x+sv}
        =
        t_x
        +
        s
        \sum_{\theta'\in\Theta_0}
        v(\theta'),
    \]
    and
    \[
        p_{x+sv}(\theta)
        =
        \frac{
            x(\theta)+sv(\theta)
        }{
            \displaystyle
            t_x
            +
            s
            \sum_{\theta'\in\Theta_0}
            v(\theta')
        }.
    \]
    We have
    \begin{align*}
        \frac{d}{ds}p_{x+sv}(\theta)
        &=
        \frac{
            \displaystyle
            v(\theta)
            \left(
                t_x
                +
                s
                \sum_{\theta'\in\Theta_0}
                v(\theta')
            \right)
        }{
            \displaystyle
            \left(
                t_x
                +
                s
                \sum_{\theta'\in\Theta_0}
                v(\theta')
            \right)^2
        }
        \\
        &\quad
        -
        \frac{
            \displaystyle
            \bigl(
                x(\theta)+sv(\theta)
            \bigr)
            \sum_{\theta'\in\Theta_0}
            v(\theta')
        }{
            \displaystyle
            \left(
                t_x
                +
                s
                \sum_{\theta'\in\Theta_0}
                v(\theta')
            \right)^2
        }.
    \end{align*}
    and hence
    \begin{align*}
        \left.
            \frac{d}{ds}p_{x+sv}(\theta)
        \right|_{s=0}
        &=
        \frac{
            \displaystyle
            v(\theta)t_x
            -
            x(\theta)
            \sum_{\theta'\in\Theta_0}
            v(\theta')
        }{
            t_x^2
        }
        \\
        &=
        \frac{
            \displaystyle
            v(\theta)
            -
            p_x(\theta)
            \sum_{\theta'\in\Theta_0}
            v(\theta')
        }{
            t_x
        }.
    \end{align*}
    Letting
    \[
        Dp_x[v](\theta)
        =
        \left.
            \frac{d}{ds}p_{x+sv}(\theta)
        \right|_{s=0},
    \]
    we have
    \begin{equation}
        \norm{Dp_x[v]}_2
        \leq
        \frac{1}{t_x}
        \sqrt{
            1
            +
            \frac{
                |\Theta_0||\Theta_1|
            }{
                |\Theta|
            }
        }
        \norm{v}_2,
    \end{equation}
    where we used the fact that, for
    $\sum_{\theta}v(\theta)=0$, we have
    \[
        \left|
            \sum_{\theta\in\Theta_0}
            v(\theta)
        \right|
        \leq
        \sqrt{
            \frac{
                |\Theta_0||\Theta_1|
            }{
                |\Theta|
            }
        }
        \norm{v}_2.
    \]
    A similar calculation for $Dq_x[v]$ yields
    \begin{equation}
        \norm{Dq_x[v]}_2
        \leq
        \frac{1}{1-t_x}
        \sqrt{
            1
            +
            \frac{
                |\Theta_0||\Theta_1|
            }{
                |\Theta|
            }
        }
        \norm{v}_2.
    \end{equation}
    Since we have $t_x\geq t_{\min}$ and
    $1-t_x\geq t_{\min}$, it follows that $\Phi^{-1}$ has Lipschitz
    constant $L$ upper bounded by
    \begin{equation}
        L
        \leq
        \sqrt{
            1
            +
            \frac{
                |\Theta_0||\Theta_1|
            }{
                |\Theta|
            }
        }
        \sqrt{
            \frac{2}{t_{\min}^2}
        }.
    \end{equation}

    Now consider the bad set defined as
    \begin{equation}
        E
        =
        \left\{
            x\in H_r\cap D
            :
            \operatorname{dist}(x,H_{r'})
            \leq
            \eta
        \right\}.
    \end{equation}
    We distinguish between two cases:

    \textbf{Case 1:}
    $H_r$ and $H_{r'}$ intersect within
    $\operatorname{affine}(\Delta(\Theta))$.

    In this case,
    \[
        H_r
        \cap
        H_{r'}
        \cap
        \operatorname{affine}(\Delta(\Theta))
    \]
    is a $|\Theta|-3$ dimensional affine subset of
    $\mathbb{R}^{|\Theta|}$. We can give the following upper bound if
    $\eta\geq d_{r,r'}$:
    \begin{align}
        \operatorname{vol}_{|\Theta|-2}(E)    
        \leq
        \frac{
            2\eta
        }{
            \sin(\alpha_{r,r'})
        }
        \sup_s
        \operatorname{vol}_{|\Theta|-3}
        \bigl(
            H_r
            \cap
            \Delta(\Theta)
            \cap
            P_s
        \bigr).
    \end{align}
    $P_s$ ranges over the affine $|\Theta|-3$ dimensional hyperplanes
    inside $H_r$ that are parallel to $H_r\cap H_{r'}$, and
    $\alpha_{r,r'}$ denotes the angle between $H_r$ and $H_{r'}$
    measured within $\operatorname{affine}(\Delta(\Theta))$. If
    $\eta<d_{r,r'}$, then
    \begin{equation}
        \operatorname{vol}_{|\Theta|-2}(E)
        =
        0.
    \end{equation}
    The supremum is upper bounded by the volume of the $|\Theta|-3$
    dimensional faces of $\Delta(\Theta)$
    \citep{alonso2023sections}:
    \[
        \sup_s
        \operatorname{vol}_{|\Theta|-3}
        \bigl(
            H_r
            \cap
            \Delta(\Theta)
            \cap
            P_s
        \bigr)
        \leq
        \frac{
            \sqrt{|\Theta|-2}
        }{
            (|\Theta|-3)!
        },
    \]
    hence
    \begin{equation}
        \operatorname{vol}_{|\Theta|-2}(E)
        \leq
        \frac{
            2\eta
        }{
            \sin(\alpha_{r,r'})
        }
        \frac{
            \sqrt{|\Theta|-2}
        }{
            (|\Theta|-3)!
        }.
    \end{equation}
    Using the Lipschitzness of $\Phi^{-1}$, we have
    \begin{equation}
        \operatorname{vol}_{|\Theta|-2}
        \bigl(
            \Phi^{-1}(E)
        \bigr)
        \leq
        L^{|\Theta|-2}
        \operatorname{vol}_{|\Theta|-2}(E).
    \end{equation}
    Finally,
    \begin{align}
        \mathbb{P}(l\cap H_r\in E)
        &=
        \frac{
            \operatorname{vol}_{|\Theta|-2}
            \bigl(
                \Phi^{-1}(E)
            \bigr)
        }{
            \operatorname{vol}_{|\Theta_0|-1}
            \bigl(
                \Delta_{\geq x_{\min}}(\Theta_0)
            \bigr)
            \operatorname{vol}_{|\Theta_1|-1}
            \bigl(
                \Delta_{\geq x_{\min}}(\Theta_1)
            \bigr)
        }
        \notag\\
        &\leq
        \frac{1}{
            C_{|\Theta_0|,|\Theta_1|}
        }
        \left(
            \frac{2}{t_{\min}^2}
        \right)^{\frac{|\Theta|-2}{2}}
        \frac{
            \eta
        }{
            \sin(\alpha_{r,r'})
        },
    \end{align}
    where we defined
    \begin{align}
        C_{|\Theta_0|,|\Theta_1|}
        :={}&
        \frac{1}{2}
        \left(
            1
            +
            \frac{
                |\Theta_0||\Theta_1|
            }{
                |\Theta|
            }
        \right)^{1-\frac{|\Theta|}{2}}
        \frac{
            (|\Theta|-3)!
        }{
            \sqrt{|\Theta|-2}
        }
        \frac{
            \sqrt{|\Theta_0||\Theta_1|}
        }{
            (|\Theta_0|-1)!
            (|\Theta_1|-1)!
        }
        \notag\\
        &\times
        \left(
            1-|\Theta_0|x_{\min}
        \right)^{|\Theta_0|-1}
        \left(
            1-|\Theta_1|x_{\min}
        \right)^{|\Theta_1|-1}.
    \end{align}

    \textbf{Case 2.}
    $H_r$ and $H_{r'}$ are parallel within
    $\operatorname{affine}(\Delta(\Theta))$.

    In this case, the distance from $l\cap H_r$ to $H_{r'}$ is
    deterministic:
    \begin{equation}
        \operatorname{dist}(l\cap H_r,H_{r'})
        \geq
        d_{r,r'}.
    \end{equation}

    It follows that to keep the overall probability below $\delta$, it
    is enough to choose
    \begin{equation}
        \eta
        =
        \max
        \left\{
            \delta
            \sin(\alpha_{r,r'})
            C_{|\Theta_0|,|\Theta_1|}
            \left(
                \frac{2}{t_{\min}^2}
            \right)^{1-\frac{|\Theta|}{2}},
            \,
            d_{r,r'}
        \right\}.
    \end{equation}
\end{proof}

The following lemma shows that $d_{r,r'}$ and $\alpha_{r,r'}$ cannot
be simultaneously too small, under Assumption \ref{as:gamma}.

\begin{lemma}
    Let $r,r'\in S^R$, $r\neq r'$ and let $H_r$ intersect the relative
    interior of $\Delta(\Theta)$. Under Assumptions \ref{as:gamma} and
    \ref{as:distinct}, we have
    \begin{equation}
        d_{r,r'}
        \geq
        \left[
            \gamma_{\min}
            \bigl(
                1+\cos(\alpha_{r,r'})
            \bigr)
            -
            \sqrt{2}
            \sin(\alpha_{r,r'})
        \right]_+.
    \end{equation}
\end{lemma}

\begin{proof}
    Let
    \[
        T
        =
        \left\{
            z
            :
            \boldsymbol{1}^\top z=0
        \right\}
    \]
    be the tangent space of $\Delta(\Theta)$, and $P_T$ denote the
    orthogonal projector onto $T$. Define
    \[
        \hat n_r
        :=
        \frac{
            P_Tn_r
        }{
            \norm{P_Tn_r}_2
        }
    \]
    be the unit orthogonal projection of $n_r$.

    We define the signed distance function (within the simplex affine
    hull ) from $H_r$ as
    \[
        d_r(\mu)
        =
        \hat n_r^\top(\mu-v_r)
    \]
    for some $v_r\in H_r\cap\Delta(\Theta)$. We define
    $d_{r'}(\mu)$ the same way for $H_{r'}$. Let
    $C\subset\Delta(\Theta)$ be a cell such that $\mathrm{sign}(d_r(\mu_C) ) =-\mathrm{sign}(d_{r'}(\mu_C) ) $, and without loss of generality, assume
    \[
        d_r(\mu)
        \geq
        0
        \qquad\text{and}\qquad
        d_{r'}(\mu)
        \leq
        0
        \qquad
        \text{for }\mu\in C,
    \]
    i.e., $C$ is a cell \emph{between} $H_r$ and $H_{r'}$. Such a cell
    exists since $H_r$ intersects
    $\operatorname{relint}(\Delta(\Theta))$. By Assumption
    \ref{as:gamma}, there exists $\mu_C\in C$ such that
    \[
        d_r(\mu_C)
        \geq
        \gamma_{\min}
        \qquad\text{and}\qquad
        d_{r'}(\mu_C)
        \leq
        -\gamma_{\min}.
    \]
    Let $\mu\in H_r\cap\Delta(\Theta)$. We have
    \begin{equation}
        d_{r'}(\mu)
        =
        d_{r'}(\mu_C)
        +
        \hat n_{r'}^\top
        (\mu-\mu_C).
    \end{equation}
    Let $v$ be a unit vector such that
    \[
        \hat n_{r'}
        =
        \cos(\alpha_{r,r'})
        \hat n_r
        +
        \sin(\alpha_{r,r'})
        v.
    \]
    Then we have
    \begin{align}
        \hat n_{r'}^\top(\mu-\mu_C)
        &=
        \cos(\alpha_{r,r'})
        \hat n_r^\top(\mu-\mu_C)
        +
        \sin(\alpha_{r,r'})
        v^\top(\mu-\mu_C)
        \notag\\
        &=
        -\cos(\alpha_{r,r'})
        d_r(\mu_C)
        +
        \sin(\alpha_{r,r'})
        v^\top(\mu-\mu_C).
    \end{align}
    Hence
    \begin{align}
        \left|
            d_{r'}(\mu)
        \right|
        &=
        \left|
            d_{r'}(\mu_C)
            -
            \cos(\alpha_{r,r'})
            d_r(\mu_C)
            +
            \sin(\alpha_{r,r'})
            v^\top(\mu-\mu_C)
        \right|
        \notag\\
        &\geq
        \bigl(
            1+\cos(\alpha_{r,r'})
        \bigr)
        \gamma_{\min}
        -
        \sqrt{2}
        \sin(\alpha_{r,r'}),
    \end{align}
    where we used that for $x,y\in\Delta(\Theta)$,
    $\norm{x-y}\leq\sqrt{2}$. By definition,
    \[
        d_{r,r'}
        =
        \min_{\mu\in H_r\cap\Delta(\Theta)}
        \left|
            d_{r'}(\mu)
        \right|.
    \]
    Since the bound is only meaningful when the right hand side is
    nonnegative, the result follows.
\end{proof}

\paragraph{Proof of Lemma~\ref{lem:separation}.}

\begin{proof}
    The case \( |\Theta|=2 \) was handled above, so suppose that
    \(n:=|\Theta|\geq3\). Let
    \[
        m
        :=
        |S^R|,
        \qquad
        M
        :=
        m(m-1),
    \]
    and define
    \[
        \kappa
        :=
        \frac{\delta}{M}
        C_{|\Theta_0|,|\Theta_1|}
        \left(
            \frac{2}{t_{\min}^2}
        \right)^{1-\frac{n}{2}}.
    \]
    Applying Lemma~\ref{lem:pairwise} with failure probability
    \(\delta/M\), and taking a union bound over the \(M\) ordered pairs
    \(r\neq r'\), shows that, with probability at least \(1-\delta\),
    simultaneously for every \(r\neq r'\) such that
    \(H_r\cap l\neq\varnothing\),
    \[
        \operatorname{dist}(l\cap H_r,H_{r'})
        \geq
        \max
        \left\{
            \kappa\sin(\alpha_{r,r'}),
            d_{r,r'}
        \right\}.
    \]
    By the angle--distance lemma,
    \[
        d_{r,r'}
        \geq
        \left[
            \gamma_{\min}
            \bigl(
                1+\cos(\alpha_{r,r'})
            \bigr)
            -
            \sqrt{2}
            \sin(\alpha_{r,r'})
        \right]_+.
    \]
    The function
    \[
        \alpha
        \longmapsto
        \kappa\sin(\alpha)
    \]
    is increasing on \([0,\pi/2]\), whereas
    \[
        \alpha
        \longmapsto
        \gamma_{\min}
        \bigl(
            1+\cos(\alpha)
        \bigr)
        -
        \sqrt{2}
        \sin(\alpha)
    \]
    is decreasing. The two expressions are equal at
    \[
        \beta
        :=
        2\arctan
        \left(
            \frac{
                \gamma_{\min}
            }{
                \sqrt{2}+\kappa
            }
        \right).
    \]
    Consequently, we may take
    \[
        \eta_{\mathrm{sep}}
        :=
        \kappa\sin(\beta).
    \]

    It remains to verify the claimed dependence of
    \(\eta_{\mathrm{sep}}\) on the problem parameters. Write
    \[
        a
        :=
        |\Theta_0|,
        \qquad
        b
        :=
        |\Theta_1|,
    \]
    so that \(a+b=n\). Since \(x_{\min}=1/(2n)\), the constant
    appearing in Lemma~\ref{lem:pairwise} is
    \begin{align}
        C_{a,b}
        ={}&
        \frac{1}{2}
        \left(
            1+\frac{ab}{n}
        \right)^{1-\frac n2}
        \frac{
            (n-3)!
        }{
            \sqrt{n-2}
        }
        \frac{
            \sqrt{ab}
        }{
            (a-1)!(b-1)!
        }
        \notag\\
        &\qquad\times
        \left(
            1-\frac{a}{2n}
        \right)^{a-1}
        \left(
            1-\frac{b}{2n}
        \right)^{b-1}.
    \end{align}
    We have
    \[
        1+\frac{ab}{n}
        \leq
        n
    \]
    and, because the exponent \(1-n/2\) is nonpositive,
    \[
        \left(
            1+\frac{ab}{n}
        \right)^{1-\frac n2}
        \geq
        n^{1-\frac n2}.
    \]
    Moreover,
    \begin{align}
        \frac{
            (n-3)!
        }{
            (a-1)!(b-1)!
        }
        \frac{
            \sqrt{ab}
        }{
            \sqrt{n-2}
        }
        &=
        \frac{1}{n-2}
        \binom{n-2}{a-1}
        \sqrt{
            \frac{ab}{n-2}
        }
        \notag\\
        &\geq
        \frac{1}{n-2},
    \end{align}
    where we used \(ab\geq n-1\). Finally,
    \[
        \left(
            1-\frac{a}{2n}
        \right)^{a-1}
        \left(
            1-\frac{b}{2n}
        \right)^{b-1}
        \geq
        2^{-(n-2)}.
    \]
    It follows that
    \[
        C_{a,b}
        \geq
        \frac{
            2^{-(n-1)}
        }{
            n-2
        }
        n^{1-\frac n2},
    \]
    and hence
    \[
        \log\frac{1}{C_{a,b}}
        =
        O(n\log n).
    \]

    Next, using
    \[
        \sin
        \bigl(
            2\arctan x
        \bigr)
        =
        \frac{
            2x
        }{
            1+x^2
        },
    \]
    we obtain
    \[
        \sin(\beta)
        =
        \frac{
            2\gamma_{\min}
            (\sqrt{2}+\kappa)
        }{
            (\sqrt{2}+\kappa)^2
            +
            \gamma_{\min}^2
        }.
    \]
    Since \(\gamma_{\min}\leq1\), there exists a universal constant
    \(c>0\) such that
    \[
        \eta_{\mathrm{sep}}
        =
        \kappa\sin(\beta)
        \geq
        c\gamma_{\min}
        \min\{\kappa,1\}.
    \]
    Therefore,
    \[
        \log\frac{1}{\eta_{\mathrm{sep}}}
        =
        O
        \left(
            \log\frac{1}{\gamma_{\min}}
            +
            \left[
                \log\frac{1}{\kappa}
            \right]_+
        \right).
    \]
    By the definition of \(\kappa\),
    \begin{align}
        \log\frac{1}{\kappa}
        &=
        \log\frac{M}{\delta}
        +
        \log\frac{1}{C_{a,b}}
        +
        \frac{n-2}{2}
        \log
        \left(
            \frac{2}{t_{\min}^2}
        \right).
    \end{align}
    Finally, since
    \[
        t_{\min}
        =
        \frac{
            d_{\min}x_{\min}z_{\min}
        }{
            2d_{\max}
        }
        =
        \frac{
            d_{\min}z_{\min}
        }{
            4nd_{\max}
        },
    \]
    we have
    \[
        \frac{n-2}{2}
        \log
        \left(
            \frac{2}{t_{\min}^2}
        \right)
        =
        O
        \left(
            n
            \log
            \frac{
                nd_{\max}
            }{
                d_{\min}z_{\min}
            }
        \right).
    \]
    Using \(M\leq m^2\), we conclude that
    \[
        \log\frac{1}{\eta_{\mathrm{sep}}}
        =
        O
        \left(
            n
            \log
            \frac{
                nd_{\max}
            }{
                d_{\min}z_{\min}
            }
            +
            \log\frac{m}{\delta}
            +
            \log\frac{1}{\gamma_{\min}}
        \right).
    \]
\end{proof}

\newpage
\section{Proof of Lemma \ref{lem:estimate}}

To prove Lemma 3, we first prove the following result, characterizing the number of trials needed so that a certain number of successes is observed with high probability, in a sequence of i.i.d. Bernoulli trials: 

\begin{lemma}{\label{lem:mult-chernoff}}
    Let $X_1, X_2,... $ be iid Bernoulli random variables with $\mathbb{E}[X_1] = p $. 
    Let $S_T = \sum_{i=1}^T X_i $.  If 
    \begin{equation}
        T \geq  \frac{1}{p}\left (N + \log \frac{1}{\delta} + \sqrt{\log^2 \frac{1}{\delta} + 2N\log \frac{1}{\delta}} \right ),
    \end{equation}
    then with probability at least $1- \delta $, $S_T \geq N $. 
\end{lemma}
\begin{proof}
    The proof simply is a consequence of the multiplicative Chernoff bound for the sum of iid Bernoulli random variables. For any $\gamma \in (0,1) $, we have
    \begin{equation}
        \mathbb{P}(S_T \leq (1- \gamma)Tp ) \leq e^{-\gamma^2Tp/2}.
    \end{equation}
    Setting $\gamma = 1 - \frac{N}{Tp} $ yields
    \begin{equation*}
        (1-\frac{N}{Tp})^2\frac{Tp}{2} \geq \log{1/\delta}.
    \end{equation*}
    Solving the resulting quadratic equation and taking the larger root results in the desired result. 
\end{proof}

\paragraph{Proof of Lemma \ref{lem:estimate}.} 

\begin{proof}
     We first provide a lower bound to the probability of jointly observing a state $ \theta $ and the signal $s_1 $.
    Let 
    \begin{equation}
        M:= \max_{\theta} \frac{\mu(\theta)}{\mu_0(\theta)} ,
    \end{equation}
     Observe that the signaling scheme used to query $\mu $, $\phi_t $, induces posterior $\mu $ with probability 
    \begin{align*}
        \mathbb{P}(s_1) = \sum_{\theta}\mu_0(\theta)\phi(s_1|\theta) &= \frac{1}{M}\sum_{\theta}\mu_0(\theta) \frac{\mu(\theta)}{\mu_0(\theta)} \\
        &= \frac{1}{M} \geq{\min_{\theta}\mu_0(\theta)}
    \end{align*}
    Hence we can lower bound the overall probability of jointly observing $s_1 $ and $\theta $ in a single round as 
    \begin{equation}        
        \mu(\theta) \min_{\theta}\mu_0(\theta) \geq \min_{\theta}\mu(\theta) \min_{\theta} \mu_0(\theta) = \alpha.
    \end{equation}
    Let $N_{\theta,T} $ denote the number of times we observe the pair $s_1, \theta $. To determine the number of rounds we need so that we observe $(s_1, \theta) $ jointly $N $ times for each $\theta $ with probability at least $1-\delta/2 $, we can choose, using Lemma \ref{lem:mult-chernoff} 
    \begin{equation}
        T \geq \frac{1}{\alpha} \left (N + \log \frac{2|\Theta|}{\delta} + \sqrt{\log^2\frac{2|\Theta|}{\delta} + 2N \log \frac{2|\Theta|}{\delta}} \right ).
    \end{equation}
    Let 
    \begin{equation}
        L = \log \frac{2|\Theta|}{\delta}, \qquad c = \frac{|\Theta|^2}{2\epsilon^2}.
    \end{equation}
    Assuming $|\Theta| \geq 2, \epsilon \leq 1 $,  we have $c \geq 2 $. Since $N \geq cL $, we have
    \begin{align}
        N + L + \sqrt{L^2 + 2NL} &\leq N \left (1 + \frac{L}{N} + \sqrt{\frac{L^2}{N^2} + \frac{2L}{N} } \right ) \\
        &\leq N \left(\frac{3}{2} + \frac{\sqrt{5}}{2} \right ) \\
        &< 2\sqrt{2}N,
    \end{align}
    and the round bound stated in Lemma \ref{lem:estimate} follows with the displayed choice of $N $ in Algorithm \ref{alg:query}.
    Finally, for this choice of $T $, using a union bound over states, we get
    \begin{align}
        \mathbb{P}\left (\exists \theta \in \Theta: N_{\theta, T} < N  \right ) &\leq \sum_{\theta \in \Theta} \mathbb{P}(N_{\theta, T} < N) \\
        &\leq \sum_{\theta} \frac{\delta}{2|\Theta|} = \frac{\delta}{2}. 
    \end{align}
    Given that we observe each state $\theta $ at least $N  $ times, using Hoeffding's inequality yields
    \begin{align}
        \mathbb{P}(|\hat p_{\mu}-p_{\mu}|_1 > \epsilon) &\leq \mathbb{P}(\exists \theta \in \Theta, |\hat p_{\mu}(a_0|\theta)-p_{\mu}(a|\theta) | > \epsilon / |\Theta|) \\
        &\leq \sum_{\theta} \frac{\delta}{2|\Theta|} = \frac{\delta}{2}.
    \end{align}
    Hence the failure probability is at most $\delta/2 + \delta/2 =\delta  $.
\end{proof}

\newpage

\section{Proof of Lemma \ref{lem:crossings}}

We first provide the pseudocode for the procedure $\textsc{Find-Crossings} $.

\begin{algorithm}
\caption{$\textsc{Find-Crossings}$} \label{alg:crossings}
\begin{algorithmic}[1]
\Require A pair of endpoints $\mu_{-} $ and $\mu_{+} $, binary search resolution $\eta $, minimum signal mass $p_{\min} $, maximum failure probability $\delta $ 
\State $\mathcal{C} \gets \emptyset $ \Comment{A cache of previously queried points}
\State $P \gets \emptyset $
\State $\mathcal{I} \gets \{[\mu_{-}, \mu_{+}] \} $
\State $L_{\eta} \gets  \left \lceil \log_2^{+} \left ( \frac{\norm{\mu_{+}-\mu_{-}}_2}{\eta} \right ) \right \rceil$
\State $Q_{\mathrm{cross}} \gets 2 + |S^R|L_{\eta} $ \Comment{Number of maximum calls to the function \textsc{Query}}

\Function{Query}{$\mu$}
    \If{$\mu \notin \mathcal{C}$}
        \State $C[\mu] \gets \textsc{Estimate-Profile}(\mu, p_{\min}/6, \delta/Q_{\mathrm{cross}}) $
    \EndIf
    \State \Return $C[\mu]$
\EndFunction

\While{$\mathcal{I} \neq \emptyset \text{ and } |C| \leq Q_{\mathrm{cross}} $}
    \State Remove an interval $[u,v] $ from $\mathcal{I} $
    \State $d(u,v) \gets \norm{\textsc{Query}(u) -\textsc{Query}(v)}_1 $
    \If{$d(u,v) < p_{\min}/2 $}
        \State \textbf{continue}
    \EndIf
    \If{$\norm{u-v} \leq \eta $}
        \State $P \gets  P \cup \{(u,v)\}$
        \State \textbf{continue}
    \EndIf
    \State $m \gets (u +v)/2 $
    \If{$\norm{\textsc{Query}(u)-\textsc{Query}(m)}_1  \geq p_{\min}/2 $}
        \State $\mathcal{I} \gets \mathcal{I} \cup \{(u,m)\} $
    \EndIf
    \If{$\norm{\textsc{Query}(m)-\textsc{Query}(v)}_1  \geq p_{\min}/2 $}
        \State $\mathcal{I} \gets \mathcal{I} \cup \{(m,v)\} $
    \EndIf
    
\EndWhile
\State \Return $P$
        
\end{algorithmic}
\end{algorithm}

\paragraph{Proof of Lemma 4.}

\begin{proof}
    Let $\mathcal{G} $ be the event that every posterior query made by the algorithm returns an estimate satisfying
    \begin{equation}
        \norm{p_{\mu} - \hat p_{\mu}}_1 \leq \epsilon = \frac{p_{\min}}{6}
    \end{equation}
    and ends within 
    \begin{equation}
        N(z_{\min}, \epsilon, \frac{\delta}{Q_{\mathrm{cross}}})
    \end{equation}
    rounds. Since the algorithm makes at most $Q_{\mathrm{cross}} $ calls, by a union bound, we have
    \begin{equation}
        \mathbb{P}(\mathcal{G}\mid \mathcal{E}_{\mathrm{sep}}) \geq 1- \delta.
    \end{equation}
    
    On $\mathcal{G} $, for any two queried points $u, v$,
    \begin{equation}
        | \norm{p_{u}-p_v}_1 - \norm{\hat p_u-\hat p_v}_1 | \leq \norm{p_u - \hat p_u}_1 + \norm{p_v - \hat p_v}_1  \leq \frac{p_{\min}}{3}.
    \end{equation}

    Let 
    \begin{equation}
        J(u, v) = \{r \in S^R, a^*(u, r) \neq a^*(v,r)\}.
    \end{equation}
    Consider an $r \in D $. Initially, we have
    \begin{equation}
        r \in J(\mu_-, \mu_+)
    \end{equation}
    and $m(r) \geq p_{\min} $. For any pairs $(u, v) $ returned by the binary search procedure, we have
    \begin{equation}
                \norm{\hat p_{u} - \hat p_v}_1 \geq \norm{p_u -p_v}_1 - \frac{p_{\min}}{3} \geq \frac{2p_{\min}}{3} > \frac{p_{\min}}{2}. 
    \end{equation}
    Hence an interval containing a jump corresponding to $r $ is never pruned by the algorithm. If $\norm{u -v}_2 > \eta $, the algorithm bisects the line segment at $m = (u + v)/2 $. If $\mu_r \neq m $, then exactly one child of the initial interval contains $\mu_r $. On the other hand, if $\mu_r = m $, then $\mu_r $ lies at the common endpoint of both children, and because of how the receiver breaks ties at $\mu_r = m $, the jump at $\mu_r $ is counted at exactly one of the children, hence a child containing $\mu_r $ is inserted into the queue. Finally, after at most $L_{\eta}  $ splits, the algorithm outputs an interval containing $r $, proving property $(i) $. 

    To prove property $(ii) $, note that every $(u,v) $ inserted into the queue satisfies
    \begin{equation}
        \norm{\hat p_{u} -\hat p_v}_1 \geq \frac{p_{\min}}{2},
    \end{equation}
    so we have
    \begin{equation} \label{eq:minmass}
        \norm{ p_u- p_v}_1 \geq \frac{p_{\min}}{2} - \frac{p_{\min}}{3} = \frac{p_{\min}}{6} > 0
    \end{equation}
    Hence $J(u,v) \neq \emptyset $. Observe that at any recursion depth, the sets $J(u,v) $ for the intervals are disjoint. This is because each signal changes the preferred action exactly once, hence a crossing at a shared endpoint is assigned to exactly one adjacent interval. It follows that there are at most $|S^R| $ active intervals at each recursion depth, giving at most  
    \begin{equation}
        2 + |S^R|L_{\eta}= Q_{\mathrm{cross}}
    \end{equation}
    queries. 

    Finally, let $(u_i, v_i) $ be an output pair. By \eqref{eq:minmass}, 
    \begin{equation}
        \norm{p_{u_i}-p_{v_i}}_1 \geq p_{\min}/6.
    \end{equation}
    Since we have $ \norm{u_i-v_i}_2 \leq \eta =  \eta_{\mathrm{sep}}/4 $, $\mathcal{E}_{\mathrm{sep}} $ implies that the interval contains at most one hyperplane.  Hence we have
    \begin{equation}
        J(u_i, v_i) = \{r(i)\}
    \end{equation}
    for a unique signal $r(i) $, with 
    \begin{equation}
        m(r(i)) \geq p_{\min}/6.
    \end{equation}
    Moreover, $a^*(u_i,r(i))=a_0 $ and $a^*(v_i, r(i)) = a_1 $, proving property $(ii) $.

    Property $(iii) $ directly follows from the stopping rule of the algorithm. 

    Since on $\mathcal{G}$, every query takes at most 
    \begin{equation}
        N(z_{\min}, p_{\min}/6 , \delta/Q_{\mathrm{cross}})
    \end{equation}
    rounds, the total number of rounds is at most 
    \begin{equation}
        Q_{\mathrm{cross}}  N(z_{\min}, p_{\min}/6 , \delta/Q_{\mathrm{cross}}).
    \end{equation}

    Finally, using $\norm{\mu_- -\mu_+}_2 \leq \sqrt{2} $, we have
    \begin{equation}
        Q_{\mathrm{cross}} \leq 2 + |S^R| \left \lceil \log_2 \frac{4 \sqrt{2}}{\eta_{\mathrm{sep}}} \right \rceil = O\left (|S^R|\log_2 \frac{4\sqrt{2}}{\eta_{\mathrm{sep}}} \right).
    \end{equation}    
\end{proof}

\newpage

\section{Proof of Lemma \ref{lem:learn-signal}}

Let
\[
    n:=|\Theta|,
    \qquad
    d:=n-1,
    \qquad
    \mathcal T:=\{x\in\mathbb R^n:\mathbf 1^\top x=0\}.
    ,
\]
and let 

\[
    \rho
    :=
    \frac{x_{\min}^2 z_{\min}d_{\min}}{2d_{\max}}.
\]
Every posterior on the search line has every coordinate at least \(\rho\).
Since decreasing the separation parameter only strengthens the event
\(\mathcal E_{\mathrm{sep}}\), we henceforth replace
\(\eta_{\mathrm{sep}}\) by
\[
    \min\{\eta_{\mathrm{sep}},\rho/4\}.
\]
This modification does not change the order of
\(\log(1/\eta_{\mathrm{sep}})\) in Lemma~\ref{lem:separation}.

Condition on \(\mathcal E_{\mathrm{sep}}\cap\mathcal E_{\mathrm{cross}}\).
For an output pair \((u,v)\in P\), let \(r\) be its corresponding
receiver signal and let \(x_r=[u,v]\cap H_r\).  Lemma~\ref{lem:crossings}
and the tie-breaking convention imply
\begin{equation}\label{eq:learn-pair-properties}
    \|u-v\|_2\leq \eta_{\mathrm{sep}}/4,
    \qquad
    n_r^\top u\geq 0 \geq n_r^\top v \qquad, a^*(u,r) \neq a^*(v,r)
    \qquad
    m(r)=\|p_u-p_v\|_1\geq p_{\min}/6.
\end{equation}
Moreover, \([u,v]\) meets no hyperplane other than \(H_r\), and hence
\begin{equation}\label{eq:profile-jump-equals-likelihood}
    p_u(\theta)-p_v(\theta)=\phi^R(r\mid\theta)
    \qquad\text{for every }\theta\in\Theta.
\end{equation}

We first analyze the two subroutines separately.

\begin{algorithm}
\caption{\textsc{Estimate-Likelihood}}\label{alg:likelihood}
\begin{algorithmic}[1]
    \Require Posteriors \(u,v\), accuracy \(\epsilon_m\), failure probability \(\delta\)
    \State \(\widehat p_u\gets
        \Call{Estimate-Profile}{u,\epsilon_m/2,\delta/2}\)
    \State \(\widehat p_v\gets
        \Call{Estimate-Profile}{v,\epsilon_m/2,\delta/2}\)
    \State \Return
    \(\widehat p:=
      \bigl(|\widehat p_u(\theta)-\widehat p_v(\theta)|\bigr)_{\theta\in\Theta}\)
\end{algorithmic}
\end{algorithm}

\begin{lemma}\label{lem:estimate-likelihood}
Under \(\mathcal E_{\mathrm{sep}}\cap\mathcal E_{\mathrm{cross}}\), let
\((u,v)\) correspond to signal \(r\).  The subroutine
\(\textsc{Estimate-Likelihood}(u,v,\epsilon_m,\delta)\) makes two calls
to \(\textsc{Estimate-Profile}\).  With probability at least
\(1-\delta\), it returns \(\widehat p\) satisfying
\[
    \|\widehat p-\phi^R(r\mid\cdot)\|_1\leq\epsilon_m.
\]
On the same event, it uses at most
\[
    2N(z_{\min},\epsilon_m/2,\delta/2)
    \leq
    8N(z_{\min},\epsilon_m,\delta/2)
\]
rounds.
\end{lemma}

\begin{proof}
By Lemma~\ref{lem:estimate} and a union bound, with probability at least
\(1-\delta\),
\[
    \|\widehat p_u-p_u\|_1\leq\epsilon_m/2
    \qquad\text{and}\qquad
    \|\widehat p_v-p_v\|_1\leq\epsilon_m/2.
\]
For vectors \(a,b\in\mathbb R^n\), the componentwise absolute-value map
is \(1\)-Lipschitz in \(\ell_1\), so, using
\eqref{eq:profile-jump-equals-likelihood},
\begin{align*}
    \|\widehat p-\phi^R(r\mid\cdot)\|_1
    &=
    \bigl\|
       |\widehat p_u-\widehat p_v|-|p_u-p_v|
     \bigr\|_1 \\
    &\leq
    \|(\widehat p_u-\widehat p_v)-(p_u-p_v)\|_1 \\
    &\leq
    \|\widehat p_u-p_u\|_1+
    \|\widehat p_v-p_v\|_1
    \leq\epsilon_m.
\end{align*}
Both \(u\) and \(v\) lie on the search line, so their coordinates are
at least \(\rho\).  The round bound follows from
Lemma~\ref{lem:estimate}.  Finally,
\(N(z_{\min},\epsilon_m/2,\delta/2)
 \leq 4N(z_{\min},\epsilon_m,\delta/2)\), including ceilings.
\end{proof}

We next give a complete version of the hyperplane-learning procedure in Algorithm \ref{alg:hyperplane}. 

We first construct the orthogonal basis for $\mathcal{T} $ as follows. First, we sample a unit vector $q_1 $ from $\mathcal{T} $. Then, after having sampled $q_1,...q_i $, we sample a unit vector $q_{i+1} $ uniformly randomly from the orthogonal complement of the preceding vectors. Randomized basis construction makes sure that, with probability 1, no search line later used by the algorithm is contained in $H_r $. This is necessary as $p_{\mu} $ can change on $H_r $ due to sender favorable tie breaking by the receiver. Let $Q = [q_1 \cdots\ q_d] $

For a queried posterior \(\mu\), let
\(\widehat p_\mu\) denote a cached call to
\(\textsc{Estimate-Profile}(\mu,p_{\min}/48,\delta/Q_h)\), and define
\begin{equation}\label{eq:side-oracle}
    \textsc{Side}(\mu)
    :=
    \begin{cases}
        +,&
        \|\widehat p_\mu-\widehat p_u\|_1
        \leq
        \|\widehat p_\mu-\widehat p_v\|_1,\\
        -,&\text{otherwise.}
    \end{cases}
\end{equation}
Here \(\widehat p_u\) and \(\widehat p_v\) are computed once at the
start.  The routine \(\textsc{Bisect}\) takes a parameterized segment
whose endpoints have opposite \(\textsc{Side}\) labels, repeatedly queries
its midpoint, and returns the midpoint parameter of a final bracket of
length at most twice its accuracy argument.  For the segment \([u,v]\),
we identify this output with the corresponding posterior.

For \(n=|\Theta|\), \(d=n-1\), and \(\epsilon_h\in(0,1]\), define
\begin{align}
    M_0&:=\left\lceil\log_2(8\sqrt d)\right\rceil,\\
    M&:=\left\lceil
       \log_2\!\left(64n\sqrt d/\epsilon_h\right)
       \right\rceil,\\
    Q_h&:=4+M_0+4d(d-1)+(2d-1)M.
    \label{eq:Qh}
\end{align}

\begin{algorithm}
\caption{\textsc{Learn-Hyperplane}}\label{alg:hyperplane}
\begin{algorithmic}[1]
    \Require Posteriors \(u,v\) satisfying \eqref{eq:learn-pair-properties}, desired normal accuracy \(\epsilon_h\), failure probability \(\delta\)
    \State \(R\gets\|u-v\|_2\); choose a random orthonormal basis \(q_1,\ldots,q_d\) of \(\mathcal T\), and set \(Q=[q_1\ \cdots\ q_d]\)
    \State \(s\gets R/(8\sqrt d)\), \(L\gets R/2\),
    \(\alpha\gets R/(16\sqrt d)\),
    \(\beta\gets s\epsilon_h/(16|\Theta|)\)
    \State Compute \(\widehat p_u,\widehat p_v\) and initialize the cache used by \(\textsc{Side}\)
    \State \(c\gets\textsc{Bisect}([u,v],\alpha)\)

    \For{\(j=1,\ldots,d\)}
        \State \(\mathrm{Good}\gets\mathrm{True}\)
        \For{\(i\in[d]\setminus\{j\}\), \(\sigma\in\{-1,+1\}\)}
            \State \(\mu_-\gets c+\sigma s q_i-Lq_j\),
            \(\mu_+\gets c+\sigma s q_i+Lq_j\)
            \If{\(\textsc{Side}(\mu_-)=\textsc{Side}(\mu_+)\)}
                \State \(\mathrm{Good}\gets\mathrm{False}\)
            \EndIf
        \EndFor
        \If{\(\mathrm{Good}\)}
            \State \(j^\star\gets j\); \textbf{break}
        \EndIf
    \EndFor

    \For{\(i\in[d]\setminus\{j^\star\}\), \(\sigma\in\{-1,+1\}\)}
        \State \(\widehat t_i^\sigma\gets
        \textsc{Bisect}(\{c+\sigma s q_i+tq_{j^\star}:t\in[-L,L]\},\beta)\)
    \EndFor

    \State \(\widehat g_{j^\star}\gets1\), and
    \(\widehat g_i\gets(\widehat t_i^- -\widehat t_i^+)/(2s)\) for \(i\neq j^\star\)
    \If{\(d=1\)}
        \State \(\widehat t_0\gets
        \textsc{Bisect}(\{c+tq_{j^\star}:t\in[-L,L]\},\beta)\),
        \(\widehat h\gets-\widehat t_0\), \(y_+\gets c+Lq_{j^\star}\)
    \Else
        \State Choose \(i_0\neq j^\star\), set
        \(\widehat h\gets-(\widehat t_{i_0}^++\widehat t_{i_0}^-)/2\),
        \(y_+\gets c+s q_{i_0}+Lq_{j^\star}\)
    \EndIf
    \State \(\widehat a\gets
       Q\widehat g+
       (\widehat h-\widehat g^\top Q^\top c)\mathbf 1\)
    \State \(\zeta\gets +1\) if \(\textsc{Side}(y_+)=+\), and \(\zeta\gets-1\) otherwise
    \State \Return \(\widehat n\gets\zeta\widehat a/\|\widehat a\|_2\)
\end{algorithmic}
\end{algorithm}

\begin{lemma}\label{lem:learn-hyperplane}
Under \(\mathcal E_{\mathrm{sep}}\cap\mathcal E_{\mathrm{cross}}\), let
\((u,v)\) correspond to signal \(r\).  Algorithm~\ref{alg:hyperplane}
makes at most
\[
    Q_h
    =O\!\left(
       |\Theta|^2+|\Theta|\log_2\!\frac{32|\Theta|}{\epsilon_h}
      \right)
\]
calls to \(\textsc{Estimate-Profile}\).  With probability at least
\(1-\delta\), it returns \(\widehat n\) such that
\[
    \|\widehat n-n_r\|_2\leq\epsilon_h.
\]
On the same event, it uses at most
\[
    2Q_hN\!\left(
      z_{\min},p_{\min}/48,\delta/Q_h
    \right)
\]
rounds.
\end{lemma}

\begin{proof}
Let \(x_r=[u,v]\cap H_r\).  Since \(R=\|u-v\|_2\leq
\eta_{\mathrm{sep}}/4\), every point of \([u,v]\) is at distance at
most \(R\) from \(x_r\), so the initial bisection crosses no hyperplane
other than \(H_r\).  Its final midpoint satisfies
\begin{equation}\label{eq:c-close-to-H}
    \operatorname{dist}_{\mathcal T}(c,H_r)\leq\alpha.
\end{equation}
Every subsequent queried posterior has the form
\(c+\sigma s q_i+tq_j\), where \(|t|\leq L\) and the term
\(\sigma s q_i\) may be absent.  Therefore
\begin{align*}
    \|c+\sigma s q_i+tq_j-x_r\|_2
    &\leq \alpha+s+L \\
    &\leq \frac{11R}{16}
    \leq \frac{11\eta_{\mathrm{sep}}}{64}
    <\eta_{\mathrm{sep}}.
\end{align*}
Thus no queried posterior crosses a hyperplane other than \(H_r\).
Furthermore, because
\(\eta_{\mathrm{sep}}\leq\rho/4\), every such posterior has each
coordinate at least
\[
    \rho-\frac{11R}{16}
    \geq \rho-\frac{11\rho}{256}
    >\rho/2.
\]

Let \(\mathcal G\) be the event that every profile query made by the
algorithm is \(p_{\min}/48\)-accurate in \(\ell_1\) and terminates
within the bound supplied by Lemma~\ref{lem:estimate}.  The query count
verified below and a union bound give
\(\Pr(\mathcal G)\geq1-\delta\).  We prove correctness conditioned on
\(\mathcal G\).

Throughout the queried region, the true profile is either \(p_u\) or
\(p_v\).  If \(p_\mu=p_u\), then
\[
    \|\widehat p_\mu-\widehat p_u\|_1
    \leq p_{\min}/24,
\]
whereas, by \eqref{eq:learn-pair-properties},
\[
    \|\widehat p_\mu-\widehat p_v\|_1
    \geq \|p_u-p_v\|_1-p_{\min}/24
    \geq p_{\min}/8.
\]
The analogous inequalities hold when \(p_\mu=p_v\).  Hence
\(\textsc{Side}\) returns the correct side of \(H_r\) for every query.

Let \(P_{\mathcal T}=QQ^\top\) be the orthogonal projector onto
\(\mathcal T\), and define
\[
    \tau:=\|P_{\mathcal T}n_r\|_2,
    \qquad
    w:=\frac{Q^\top n_r}{\tau}\in\mathbb S^{d-1},
    \qquad
    b:=\frac{n_r^\top c}{\tau}.
\]
Here \(\tau>0\), because \(H_r\) intersects the relative interior of
the simplex.  In the local coordinates \(\mu=c+Qx\), the hyperplane is
\begin{equation}\label{eq:local-hyperplane}
    b+w^\top x=0.
\end{equation}
Because \(w\) is a unit vector and by \eqref{eq:c-close-to-H},
\(|b|\leq\alpha\).

There exists a coordinate \(j\) with \(|w_j|\geq1/\sqrt d\).  For this
coordinate, every \(i\neq j\) and \(\sigma\in\{-1,+1\}\) gives a
crossing at
\[
    t_i^\sigma=-\frac{b+\sigma s w_i}{w_j},
\]
and
\[
    |t_i^\sigma|
    \leq \sqrt d\,\alpha+\sqrt d\,s
    =\frac{3R}{16}<L.
\]
Thus this \(j\) passes the test, so the algorithm finds some
\(j^\star\).  For the accepted coordinate, every tested offset line has
a unique threshold in \([-L,L]\).  When \(d=1\), the central line has a
threshold in \((-L,L)\), because
\(|b/w_{j^\star}|\leq\alpha<L\).

Set
\[
    g:=\frac{w}{w_{j^\star}},
    \qquad
    h:=\frac{b}{w_{j^\star}}.
\]
Then \(g_{j^\star}=1\), the central threshold is \(t_0=-h\), and the
offset thresholds satisfy
\[
    t_i^\sigma=-h-\sigma s g_i.
\]
Every final bisection interval has length at most \(2\beta\); hence its
midpoint is within \(\beta\) of the true threshold.  If \(d\geq2\),
\(h=-(t_{i_0}^++t_{i_0}^-)/2\); if \(d=1\), \(h=-t_0\).  Consequently,
\begin{equation}\label{eq:g-h-errors}
    |\widehat h-h|\leq\beta,
    \qquad
    |\widehat g_i-g_i|\leq\frac{\beta}{s}
    \quad(i\neq j^\star),
\end{equation}
and therefore
\(\|\widehat g-g\|_2\leq\sqrt d\,\beta/s\).

Define
\[
    a:=Qg+(h-g^\top Q^\top c)\mathbf 1.
\]
For every \(\mu=c+Qx\) in the affine hull of the simplex,
\[
    a^\top\mu=h+g^\top x
    =\frac{n_r^\top\mu}{\tau w_{j^\star}}.
\]
Equivalently, using \(QQ^\top=P_{\mathcal T}\) and
\((I-P_{\mathcal T})c=\mathbf 1/n\),
\[
    Qg=\frac{P_{\mathcal T}n_r}{\tau w_{j^\star}},
    \qquad
    (h-g^\top Q^\top c)\mathbf 1
    =\frac{(I-P_{\mathcal T})n_r}{\tau w_{j^\star}}.
\]
Hence
\begin{equation}\label{eq:a-equals-normal}
    a=\frac{n_r}{\tau w_{j^\star}}.
\end{equation}
In particular, \(\|a\|_2\geq1\).  Since \(\|Q^\top c\|_2\leq1\),
\(\|\mathbf 1\|_2=\sqrt n\), and \(s\leq1\),
\begin{align*}
    \|\widehat a-a\|_2
    &\leq
      \|\widehat g-g\|_2
      +\sqrt n\bigl(
          |\widehat h-h|
          +\|\widehat g-g\|_2
       \bigr) \\
    &\leq
      \bigl((1+\sqrt n)\sqrt d+\sqrt n\bigr)\frac{\beta}{s} \\
    &\leq 4n\frac{\beta}{s}
     =\frac{\epsilon_h}{4}.
\end{align*}
Because \(\epsilon_h\leq1\), the preceding bound and
\(\|a\|_2\geq1\) imply \(\|\widehat a\|_2\geq3/4\).  Therefore
\[
    \left\|
       \frac{\widehat a}{\|\widehat a\|_2}
       -\frac{a}{\|a\|_2}
    \right\|_2
    \leq
    \frac{2\|\widehat a-a\|_2}{\|a\|_2}
    \leq\epsilon_h/2.
\]
The orientation witness \(y_+\) is the \(+L\) endpoint of a line whose
endpoints have opposite side labels, and the affine functional \(a\)
increases with its \(t\)-parameter.  Hence
\(\textsc{Side}(y_+)=+\) precisely when \(n_r\) is a positive multiple
of \(a\); otherwise it is a negative multiple of \(a\).  The factor
\(\zeta\) thus selects the orientation in
\eqref{eq:a-equals-normal} that agrees with $a^*(u, r) = a_0 $ and $a^*(v, r = a_1) $.  Hence the returned vector satisfies
\(\|\widehat n-n_r\|_2\leq\epsilon_h\).

It remains to count queries.  The initial bisection uses at most
\[
    \left\lceil
       \log_2\frac{R}{2\alpha}
    \right\rceil
    =M_0
\]
new midpoint queries.  The pivot test uses at most \(4d(d-1)\) profile
queries.  When \(d=1\), the two central endpoints require two more queries.
The algorithm bisects \(2(d-1)\) offset lines when \(d\geq2\), and one
central line when \(d=1\); in either case there are at most \(2d-1\)
threshold searches, each requiring at most
\[
    \left\lceil\log_2\frac{L}{\beta}\right\rceil
    =M
\]
new midpoint queries.  The offset endpoints are already cached from the
pivot test.  Together with the two reference profiles, the total is at
most \(Q_h\).  Since
\[
    \frac{L}{\beta}
    =\frac{64n\sqrt d}{\epsilon_h},
\]
we obtain
\[
    Q_h
    =O\!\left(
       n^2+n\log_2\frac{32n}{\epsilon_h}
      \right).
\]
Finally, all locally queried posteriors have minimum coordinate at least
\(\rho/2\).  Lemma~\ref{lem:estimate} therefore bounds each query by
\(N(z_{\min},p_{\min}/48,\delta/Q_h)\) rounds.  Multiplying by
\(Q_h\) proves the stated round bound.
\end{proof}

\paragraph{Proof of Lemma~\ref{lem:learn-signal}.}
\begin{proof}
For each \(i\in[k]\), apply Lemma~\ref{lem:estimate-likelihood} with
failure probability \(\delta/(2k)\), and apply
Lemma~\ref{lem:learn-hyperplane} with failure probability
\(\delta/(2k)\).  A union bound over the \(2k\) subroutine calls shows
that all likelihood and normal estimates are simultaneously correct
with probability at least \(1-\delta\).  On this event, for every
\(i\in[k]\),
\[
    \|\widehat p_i-\phi^R(r(i)\mid\cdot)\|_1\leq\epsilon_m,
    \qquad
    \|\widehat n_i-n_{r(i)}\|_2\leq\epsilon_h.
\]

The intervals output by \(\textsc{Find-Crossings}\) correspond to
distinct crossings: two different terminal intervals cannot contain
the same crossing, and the tie-breaking convention assigns a crossing
at a shared endpoint to exactly one child.  Hence \(k\leq|S^R|\).
Let \(Q_h\) be the bound in Lemma~\ref{lem:learn-hyperplane} and define
the padded query budget
\begin{equation}\label{eq:Q-learn-explicit}
    Q_{\mathrm{learn}}:=8k+2kQ_h.
\end{equation}
The actual number of calls to \(\textsc{Estimate-Profile}\) is only
\(2k+kQ_h\), and therefore is at most \(Q_{\mathrm{learn}}\).  Moreover,
using \(k\leq|S^R|\),
\[
    Q_{\mathrm{learn}}
    =O\!\left(
       |\Theta|^2|S^R|
       +|\Theta||S^R|
        \log_2\frac{32|\Theta|}{\epsilon_h}
      \right).
\]

Set
\[
    \epsilon_*:=\min\{\epsilon_m,p_{\min}/48\}.
\]
The likelihood-estimation calls use at most
\begin{align*}
    2kN\!\left(
       z_{\min},\epsilon_m/2,\frac{\delta}{4k}
    \right)
    &\leq
    8kN\!\left(
       z_{\min},\epsilon_*,\frac{\delta}{Q_{\mathrm{learn}}}
    \right)
\end{align*}
rounds.  The hyperplane-estimation calls use at most
\begin{align*}
    2kQ_hN\!\left(
       z_{\min},p_{\min}/48,
       \frac{\delta}{2kQ_h}
    \right)
    &\leq
    2kQ_hN\!\left(
       z_{\min},\epsilon_*,
       \frac{\delta}{Q_{\mathrm{learn}}}
    \right)
\end{align*}
rounds.  Here we used the monotonicity of \(N\) in its accuracy and
failure-probability arguments, together with
\(Q_{\mathrm{learn}}\geq8k\) and
\(Q_{\mathrm{learn}}\geq2kQ_h\).  Summing the two displays and using
\eqref{eq:Q-learn-explicit} gives
\[
    Q_{\mathrm{learn}}
    N\!\left(
       z_{\min},
       \min\{\epsilon_m,p_{\min}/48\},
       \frac{\delta}{Q_{\mathrm{learn}}}
    \right),
\]
as claimed.
\end{proof}

\newpage

\section{Proof of Lemma \ref{lem:subopt}}

Before we prove Lemma \ref{lem:subopt}, we first provide the pseudocode for the subroutine $\textsc{Compute-Signaling} $.

\begin{algorithm}
\caption{\textsc{Compute-Signaling}}\label{alg:commit}
\begin{algorithmic}[1]
        \Require An estimate $\hat \phi^R $ of $\phi^R $ with alphabet $\hat S^R $, special signals $r^0 $ and $r^1 $
        \State Set the signaling alphabet as $S \cup \{s^b\} $ with $|S| = |\Theta| $
        \State For all $y \in \{0,1\}^{{(|\hat S^R|+2})(|S|+1)} $ such that $y_{s, r^0}=0 $ and $y_{s,r^1} =1 $, solve the following LP:
        \begin{subequations}
\begin{align}
\max_{\phi}\quad
& \sum_{\theta \in \Theta}\sum_{s \in S \cup \{s^b\} } \sum_{a \in A} \sum_{r \in\hat S^R \cup \{r^0, r^1\} } 
  \mu_0(\theta)\phi_\theta(s)
  \hat\phi^R_\theta(r)
  u_\theta(a_{y_{s,r}} )
\\
\text{s.t.}\quad
& \phi_\theta \in \Delta(S \cup \{s^b\}),
\quad  \theta\in\Theta,
\\
& (-1)^{y_{s,r }}\sum_{\theta\in\Theta}
  \mu_0(\theta)\phi_{\theta}(s) \hat n_r(\theta) \geq 2\epsilon_h\sum_{\theta}\mu_0(\theta)\phi_{\theta}(s)  ,\quad (s,r)\in S\times \hat  S^R ,
\\
& \sum_{\theta\in\Theta}
  \mu_0(\theta)\phi_\theta(s^b)
  \le
  \frac{2\sqrt{2}\epsilon_h}
       {\gamma_{\min}\min_{\theta\in\Theta}\mu_0(\theta)} .
\end{align}
\end{subequations}
\State Let $\hat \phi $ be the signaling scheme achieving the highest utility among the at most $2^{{(|\hat S^R|+2})(|S|+1)} $ computed schemes
\State \Return $\hat \phi $
        
\end{algorithmic}
\end{algorithm}

\subsection{Proof of Lemma \ref{lem:subopt}}

\begin{proof}
    Let $V^* $ denote the value of the original problem, achieved by a signaling scheme $\phi^* $ using $|\Theta| $ signals. We have
    \begin{equation}
        V^* = \sum_{s \in S} \pi_s V(\mu_s).
    \end{equation}
    
    \paragraph{Cost of Robustness. } To characterize the loss in utility due to the margin in \eqref{eq:compute-threshold}, we first construct a \emph{robust} signaling scheme $\phi^*_{\mathrm{rob}} $ that is feasible for \eqref{eq:compute}, as follows:
    \begin{enumerate}
        \item For all $\mu_s $ induced by $\phi $ where $\mu_s \in C_s $, find $\mu_{C_s} \in C_s $ that satisfy, for all $r \in S^R $ 
        \begin{equation}
            |n_r^T\mu_{C_s}| \geq \gamma_{\min}.
        \end{equation}
        Note that $\mu_{C_s} $ exists due to Assumption \ref{as:gamma}. 
        \item For all $s $, compute $\mu_s' := (1-\lambda)\mu_s + \lambda \mu_{C_s} $, where $\lambda := \frac{4 \epsilon_h}{\gamma_{\min}} $.
        \item Let $\mu_0' = \sum_s \pi_s \mu_s' $. To make the signaling scheme Bayes-plausible, let $t_s > 1 $ be such that  $\mu' := t_s \mu_0 + (1-t_s)\mu_0' \in \partial(\Delta(\Theta))$, where $\partial $ denotes the set boundary. Let
        \begin{equation}
            \delta := \frac{\norm{\mu_0 -\mu_0'}_2}{\norm{\mu_0'-\mu'}_2} \leq \frac{\sum_s \pi_s \norm{\mu_s -\mu_s'}}{\norm{\mu_0-\mu'}} \leq \frac{\lambda \sqrt{2}}{\min_{\theta} \mu_0(\theta)}.
        \end{equation}
        We define $\phi^*_{\mathrm{rob}} $ as the signaling scheme that induces posteriors $\{\mu_s'\}_{s \in S} $ with probability $(1-\delta)\pi_s $, and induces $\mu' $ with probability $\delta $.
    \end{enumerate}
    Next, we show that $\phi^*_{\mathrm{rob}} $ is feasible for \eqref{eq:compute}, on the event $\mathcal{E}_{\mathrm{sep}} \cap \mathcal{E}_{\mathrm{cross}} \cap\mathcal{E}_{\mathrm{learn}} $. First, remember that every $\hat r \in \hat S^R $ satisfies $\norm{\hat n_r - n_r}_2 \leq \epsilon_h $. It follows that
    \begin{align}
        |\hat n_r^T \mu_s'| &= \left|(\hat n_r -n_r)^T\mu_s' + n_r^T\mu_s' \right| \\
        &\geq -\epsilon_h + \lambda \gamma_{\min} \geq 2\epsilon_h,
    \end{align}
    hence $\phi^*_{\mathrm{rob}} $ satisfies the constraint \eqref{eq:compute-threshold} for all $s \in S $, given that $y $ is chosen appropriately.

    Finally $\phi^*_{\mathrm{rob}} $ induces an additional posterior $\mu' $ that need not be robust with probability at most  $\frac{\lambda \sqrt{2}}{\min_{\theta}\mu_0(\theta)} $, meaning that $\phi^*_{\mathrm{rob}} $ also satisfies the constraint \eqref{eq:compute-bound}.

    Next, we show that $\phi^*_{\mathrm{rob}} $ is near-optimal for \eqref{eq:persuasion}. We have
    \begin{align}
        V(\mu_s)-V(\mu_s') &= \sum_{\theta}\sum_a (\mu_s(\theta)- \mu_s'(\theta))p_{\mu_s}(a|\theta)u(a,\theta) \\
        &\leq \sum_{\theta,a} p_{\mu_s}(a|\theta)\left|\mu_s(\theta)-\mu_s'(\theta)\right|u(a,\theta) \\
        &\leq \norm{\mu_s-\mu_s'}_1 \\
        &\leq \lambda \norm{\mu_s-\mu_{C_s}}_1 \leq 2\lambda = \frac{8\epsilon_h}{\gamma_{\min}}.
    \end{align}
    It follows that in the original problem \eqref{eq:persuasion}, $\phi^*_{\mathrm{rob}} $ achieves a value of at least
    \begin{align}
        V^*_{\mathrm{rob}} &\geq \left (1 -\frac{4 \sqrt{2} \epsilon_h}{\gamma_{\min} \min_{\theta}\mu_0(\theta) } \right )\left (V^* - \frac{8 \epsilon_h}{\gamma_{\min}} \right ) \\
        &\geq V^* - \frac{4 \sqrt{2} \epsilon_h}{\gamma_{\min} \min_{\theta}\mu_0(\theta) }- \frac{8 \epsilon_h}{\gamma_{\min}}.
    \end{align}

    Since on the event $\mathcal{E}_{\mathrm{sep}} \cap \mathcal{E}_{\mathrm{learn}} \cap \mathcal{E}_{\mathrm{cross}} $, every $\hat r \in \hat S^R $ correspond to a unique $r \in S^R $, by an abuse of notation, we associate each $\hat r \in \hat S^R $ with its corresponding $r \in S^R $, and consider $\hat S^R $ to be a subset of $S^R $. 

    \paragraph{Estimation Errors.} To characterize the utility achieved by $\phi^*_{\mathrm{rob}} $ in the estimated problem \eqref{eq:compute}, we compare the value of a posterior $\mu $ in \eqref{eq:persuasion}, denoted by $V(\mu) $ and in \eqref{eq:compute}, denoted by $\widehat V(\mu) $. The difference $|V(\mu)-\widehat V(\mu)| $ can be controlled as
        \begin{align}
        \left |V(\mu) - \widehat V(\mu) \right | &= \left |\sum_{\theta}\sum_a \mu(\theta)\left ( p_{\mu}(a|\theta)-\hat p_{\mu}(a|\theta) \right ) u(a,\theta) \right | \\
        &\leq  \sum_{\theta} \mu(\theta) \sum_a \left |p_{\mu}(a|\theta)-\widehat p_{\mu}(a|\theta) \right  |,
    \end{align}
    where $\hat p_{\mu} $ denotes the sender's estimate of $p_{\mu} $, computed using $(\hat \phi^R, r^0, r^1) $ as follows:
    \begin{equation}
        \widehat p_{\mu}(\theta) = r^0(\theta) +  \sum_{\hat r \in \hat S^R: a^*(\mu, \hat r)=a_0} \hat \phi^R(\hat r| \theta).
    \end{equation}
    Since on the event $\mathcal{E}_{\mathrm{sep}} \cap \mathcal{E}_{\mathrm{learn}} \cap \mathcal{E}_{\mathrm{cross}} $, every $\hat r \in \hat S^R $ correspond to a unique $r \in S^R $ (Lemma \ref{lem:crossings}), if $\mu $ satisfies \eqref{eq:compute-threshold}, we have $a^*(\mu, r) = a^*(\mu, \hat r) $ for all $\hat r \in \hat S^R $. As each $\hat r $ satisfies $\norm{\phi_r -\hat \phi_{\hat r}}_1 \leq \epsilon_m $, the contribution of all $\hat r \in \hat S^R $ to the error in $\widehat p_{\mu} $ is at most $2\epsilon_m |S^R| $, where the factor $2 $ appears as there are 2 actions. 

    Let 
        \begin{align}
        &R_{\mathrm{small}} = \left \{r \in S^R \setminus \hat S^R: \sum_{\theta}\phi^R(r |\theta) < p_{\min} \right  \} \\
        &R_0 = \left \{r \in S^R \setminus \hat S^R: \frac{\sum_{\theta \in \Theta_1}\phi^R(r |\theta)}{\sum_{\theta \in \Theta}\phi^R(r|\theta)} < z_{\min} \right\}  \\
        &R_1 = \left \{r \in S^R \setminus \hat S^R: \frac{\sum_{\theta \in \Theta_0}\phi^R(r |\theta)}{\sum_{\theta \in \Theta}\phi^R(r|\theta)} < z_{\min} \right \}.  \\
    \end{align}
    Remember that the algorithm does not distinguish between the signals in $R^0 $ but aggregates them into a special signal $r^0 $ which causes the receiver to always take the action $a_0 $. For any $r \in R^0 $, this approximation causes a (classification) error in the region given by 
        \begin{equation}
        W_r =\left \{\mu \in \Delta(\Theta): \sum_{\theta \in \Theta_1}\mu(\theta)\phi^R(r|\theta)|D_{\theta}| \geq \sum_{\theta \in \Theta_0 }\mu(\theta) \phi^R(r|\theta) |D_{\theta}| \right \}.
    \end{equation}
        If $\mu \notin W^r $, then the error in $|V(\mu) - \widehat V(\mu)| $ is at most $\epsilon_m $. On the other hand, if $\mu \in W_r $, we have 
        \begin{equation}
             \sum_{\theta \in \Theta}\mu(\theta)\phi^R(r|\theta) \leq (1 + \frac{d_{\max}}{d_{\min}})\sum_{\theta \in \Theta_1} \mu(\theta)\phi^R(r|\theta)
        \end{equation}
        
        and the contribution of $r $ to the error $|V(\mu) - \widehat V(\mu)| $ can be bounded as
    \begin{align}
        \sum_{\theta \in \Theta} \mu(\theta) \sum_{a \in A} \phi^R(r|\theta)    &= \sum_{\theta} 2\mu(\theta) \phi^R(r|\theta) \\
        &\leq 2(1 + \frac{d_{\max}}{d_{\min}}) \sum_{\theta \in \Theta_1} \mu(\theta)\phi^R(r|\theta) \\
    \end{align}
    It follows that the aggregate error caused by all signals in $R_0$ can be controlled by:
    \begin{align}
        2(1 + \frac{d_{\max}}{d_{\min}})\sum_{r \in R_0}\sum_{\theta \in \Theta_1}\mu(\theta)\phi^R(r|\theta) &\leq 2(1 + \frac{d_{\max}}{d_{\min}})z_{\min} \sum_{r \in R_0} \sum_{\theta \in \Theta} \phi^R(r|\theta) \\
        &\leq 2 (1 + \frac{d_{\max}}{d_{\min}})|\Theta|z_{\min}.
    \end{align}
    The same arguments apply to the signals in $R_1 $. Hence the error in $V(\mu) $ caused by the signals in $R_0 \cup R_1 $ is at most 
    \begin{equation}
        4(1 + \frac{d_{\max}}{d_{\min}})|\Theta|z_{\min}.
    \end{equation}
    
    The error caused by signals in $R_{\mathrm{small}} $ can be controlled as
    \begin{align}
        \sum_{\theta }\mu(\theta) \sum_{a} \left | p_{\mu}(a|\theta)-\hat p_{\mu}(a|\theta) \right |  &\leq \sum_{\theta}\sum_a\sum_{r \in R_{\mathrm{small}}} \mu(\theta)\phi^R(r|\theta) \\
        &\leq 2|S^R| p_{\min}.   
    \end{align}

    It follows that $\phi^*_{\mathrm{rob}} $ achieves in the estimated problem \eqref{eq:compute} a value of at least
    \begin{equation}
        \widehat V(\phi^*_{\mathrm{rob}}) \geq V^* -\frac{4\sqrt{2}\epsilon_h}{\gamma_{\min}\min_{\theta}\mu_0(\theta)}-\frac{8 \epsilon_h}{\gamma_{\min}} - 4(1 + \frac{d_{\max}}{d_{\min}})|\Theta|z_{\min} - 2|S^R|p_{\min}-2|S^R|\epsilon_m,
    \end{equation}
    which also provides a lower bound to the value of \eqref{eq:compute}. 

    By the same arguments, if one takes an optimal scheme for \eqref{eq:compute} and use it in the original problem \eqref{eq:persuasion},  the loss in utility due to estimation errors is at most
    \begin{equation}
        4(1 + \frac{d_{\max}}{d_{\min}})|\Theta|z_{\min} + 2|S^R| p_{\min} + 2|S^R|\epsilon_m.
    \end{equation}
    It follows that the signaling scheme $\hat \phi $ computed by the sender achieves a value of at least 
    \begin{equation}
        V^* -\frac{4\sqrt{2}\epsilon_h}{\gamma_{\min}\min_{\theta}\mu_0(\theta)}-\frac{16 \epsilon_h}{\gamma_{\min}} - 8(1 + \frac{d_{\max}}{d_{\min}})|\Theta|z_{\min} - 4|S^R| p_{\min}-4|S^R|\epsilon_m
    \end{equation}
    in the original problem \eqref{eq:persuasion}.

\end{proof}

\end{document}